\documentclass[letterpaper]{article} 
\usepackage{aaai2026}  
\usepackage{times}  
\usepackage{helvet}  
\usepackage{courier}  
\usepackage[hyphens]{url}  
\usepackage{graphicx} 
\usepackage{natbib}  
\usepackage{caption} 
\newcommand{\abox}{\mathcal{A}}
\newcommand{\ALC}{\mathcal{ALC}}
\newcommand{\ALCI}{\mathcal{ALCI}}

\newcommand{\SHIQ}{\mathcal{SHIQ}}
\newcommand{\SHQ}{\mathcal{SHQ}}

\newcommand{\SH}{\ensuremath{\mathcal{SH}}\xspace}
\newcommand{\Hmc}{\ensuremath{\mathcal{H}}\xspace}
\newcommand{\Tmc}{\ensuremath{\mathcal{T}}\xspace}
\newcommand{\Amc}{\ensuremath{\mathcal{A}}\xspace}

\newcommand{\Imc}{\ensuremath{\mathcal{I}}\xspace}
\newcommand{\Mmc}{\ensuremath{\mathcal{M}}\xspace}
\newcommand{\indiv}{\NI}

\newcommand{\interp}{\mathcal{I}}

\newcommand{\Jmc}{\mathcal{J}}

\newcommand{\Kmc}{\mathcal{K}}

\newcommand{\Mos}{\mathfrak{M}}

\newcommand{\nat}{\mathbb{N}}
\newcommand{\NV}{\mathsf{N}_{\mathsf{V}}}
\newcommand{\NC}{\mathsf{N}_{\mathsf{C}}}
\newcommand{\NI}{\mathsf{N}_{\mathsf{I}}}
\newcommand{\NR}{\mathsf{N}_{\mathsf{R}}}
\newcommand{\NRt}{\NR^{\mathsf{t}}}
\newcommand{\NRn}{\NR^{\mathsf{nt}}}

\newcommand{\pair}[2]{\langle #1, #2\rangle}

\newcommand{\qbf}{\mathbf{q}}

\newcommand{\rch}{\mathsf{VC}}
\newcommand{\rchind}{M}

\newcommand{\Smc}{\mathcal{S}}

\newcommand{\tbox}{\mathcal{T}}
\newcommand{\tp}{\mathrm{tp}}

\newcommand{\var}{\mathrm{var}}

\newcommand{\ExpTime}{\textnormal{\sc ExpTime}\xspace}
\newcommand{\NExpTime}{\textnormal{\sc NExpTime}\xspace}
\newcommand{\coNExpTime}{\textnormal{\sc coNExpTime}\xspace}
\newcommand{\TwoExpTime}{\textnormal{\sc 2ExpTime}\xspace}

\newcommand{\red}[1]{\textcolor{red}{#1}}
\renewcommand{\red}[1]{#1}

\newcommand{\new}[1]{\textcolor{blue}{#1}}
\renewcommand{\new}[1]{#1}

\newcommand{\desc}[1]{t^{\interp}(#1)}

\usepackage{soul}
\usepackage[utf8]{inputenc}

\usepackage{amsmath, amsthm, amssymb}
\usepackage{microtype}
\usepackage[switch]{lineno}
\usepackage{xspace}
\usepackage{enumitem}

\usepackage{tikz}
\usepackage{color}
\usepackage{xcolor}
\usepackage{thm-restate}
\usepackage{stmaryrd}
\usepackage{mathtools}

\newtheorem{theorem}{Theorem}

\newtheorem{claim}{Claim}

\newtheorem{corollary}[theorem]{Corollary}

\newtheorem{proposition}[theorem]{Proposition}
\newtheorem{lemma}[theorem]{Lemma}
\newtheorem{definition}[theorem]{Definition}

\title{Revisiting Conjunctive Query Entailment for $\Smc$}
\author {
    Yazm{\'{\i}}n Ib{\'{a}}{\~{n}}ez{-}Garc{\'{\i}}a\textsuperscript{\rm 1},
    Jean Christoph Jung\textsuperscript{\rm 2},
    Vincent Michielini\textsuperscript{\rm 3},
    Filip Murlak\textsuperscript{\rm 3}
}
\affiliations {
    \textsuperscript{\rm 1}Cardiff University\\
    \textsuperscript{\rm 2}TU Dortmund University\\
    \textsuperscript{\rm 3}University of Warsaw\\
ibanezgarciay@cardiff.ac.uk, jean.jung@tu-dortmund.de,
\{michielini,f.murlak\}@mimuw.edu.pl
}

\begin{document}

\maketitle

\begin{abstract}

We clarify the complexity of answering unions of conjunctive queries over knowledge bases formulated in the description logic $\Smc$, the extension of $\ALC$ with transitive roles. Contrary to what existing partial results suggested, we show that the problem is in fact \TwoExpTime-complete; hardness already holds in the presence of two transitive roles and for Boolean conjunctive queries. We complement this result by showing that the problem remains in \coNExpTime when the input query is rooted or is restricted to use at most one transitive role (but may use arbitrarily many non-transitive roles).
\end{abstract}

\section{Introduction}

In this paper, we aim to complete the complexity landscape for the problem of query answering over knowledge bases  expressed in expressive description logics (DLs), that is, logics extending the basic DL $\ALC$. As is common in such endeavor, we focus on the associated decision problem known as \emph{query entailment}: given a DL knowledge base (KB) consisting of an ABox and a TBox, a query, and a tuple of individuals, the goal is to determine whether the query returns this tuple in every model of the given KB. As query languages we consider \emph{unions of conjunctive queries (UCQs)} and fragments thereof. This problem has been heavily studied and is well understood. Existing results suggest a dichotomy: 
\begin{itemize}

\item The problem is \ExpTime-complete for $\ALC$ and its extensions with role hierarchies ($\Hmc$) or qualified number restrictions ($\mathcal Q$)~\cite{Lutz08,OrtizSE08}.

\item The problem is \TwoExpTime-complete for the extension $\ALC_\textit{self}$ of $\ALC$ with the \emph{self}-constructor~\cite{DBLP:journals/jair/BednarczykR23}, and for every extension that allows either inverse roles ($\Imc$) or both 
$\Hmc$ and transitive roles ($\Smc$), as long as it remains inside the DL $\SHIQ$, which combines the extensions $\Smc,\Hmc,\Imc,\mathcal Q$~\cite{EiterLOS09,GlimmHS08}.  
\end{itemize}
In many cases, the restriction to \emph{rooted} queries, that is, connected queries with answer variables, leads to lower complexity. This is, e.g., the case for $\ALC_\textit{self}$~\cite{DBLP:conf/ijcai/Bednarczyk24} and all DLs between $\ALCI$ and $\SHIQ$, where rooted query entailment is \coNExpTime-complete as long as transitive role names are disallowed in queries~\cite{Lutz08}. Without inverses, the complexity is even lower:
entailment of (rooted or not) queries without transitive roles is \ExpTime-complete for all DLs between $\ALC$ and $\SHQ$~\cite{Lutz08}. Working with rooted queries does not necessarily help if transitive roles are allowed in queries: analyzing the \TwoExpTime-hardness proof for query entailment in $\SH$ from~\cite{EiterLOS09} shows that rooted query entailment remains \TwoExpTime-complete.

A notable gap remains for $\Smc$ without any extensions. For $\Smc$ the problem has been proven to be $\coNExpTime$-complete when KBs and queries are allowed to use a single transitive role (and no other roles). The lower bound comes from~\cite{EiterLOS09} and the upper one from~\cite{BienvenuELOS10}. A \TwoExpTime-upper bound for the general case follows from the mentioned result for $\SH$~\cite{EiterLOS09}. In the same paper, \citeauthor{EiterLOS09}{} also claim an 
$\ExpTime$ upper bound for the case where the ABox is tree-shaped. 
These results have been regarded as a strong indication that the whole problem may be $\coNExpTime$-complete, challenging the apparent dichotomy~\cite{BienvenuELOS10}. The complexity for rooted query entailment was open until now, but a \coNExpTime-lower bound follows from~\cite{EiterLOS09}.

The aim of this paper is to revisit the query entailment problem for $\Smc$ and close the mentioned gaps. Our first main result is that, surprisingly, when at least two transitive roles are allowed in the query, the problem is \TwoExpTime-hard already for $\Smc$ (without role hierarchies), and hence \TwoExpTime-complete.
Importantly, our lower bound works with a tree-shaped ABox and thus contradicts the mentioned
\ExpTime upper bound for that case~\cite{EiterLOS09}. Indeed, in their argument, \citeauthor{EiterLOS09}{} compile input conjunctive queries (CQs) to so-called \emph{pseudo-tree queries (PTQs)}, designed to capture the behavior of CQs over tree-like interpretations, which is sufficient due to the tree-like model property of $\Smc$.
\new{The error seems to arise from a subtle mismatch between the formal definition of PTQs used in the compilation process, and their intuitive understanding as \emph{trees of clusters}, on which the later algorithmic treatment relies.} 


Besides the presence of two transitive roles in the query, our \TwoExpTime-hardness proof relies on the availability of non-rooted queries.
We complement our lower bound by showing that both conditions are necessary:  UCQ entailment remains in \coNExpTime if at most one transitive role is allowed in the query or if the query is rooted. 

We develop the two proofs in parallel, using common terminology and data structures whenever possible. 
At the core, we show a \emph{small witness property} in the following sense: the query is not entailed iff there is a small structure witnessing that. Note that this structure cannot be simply a countermodel since UCQ entailment for $\Smc$ is not finitely controllable~\cite{ROSATI2011572}. The \NExpTime-algorithm for \emph{non-}entailment can then just guess a small structure and verify that it is indeed a witness. Our argument is divided into three steps.
\begin{description}

    \item[Step~1.] \label{st:one} Reduce UCQ entailment to a special case with trivial ABoxes and UCQs of special shape.
    
    \item[Step~2.] Reduce the special case of entailment to the existence of structures called \emph{mosaics}, which are collections of interpretations called \emph{tiles} that respect certain compatibility requirements.

    \item[Step~3.]\label{st:last} Show that both the size of all tiles in a mosaic and their number can be bounded exponentially.
  
\end{description}

While the overall strategy is familiar, the steps are subtle. Steps~1--2 for the single transitive role case rely on a careful refinement of \citeauthor{EiterLOS09}{}'s PTQs, \new{which ensures correctness of the algorithmic treatment, but makes the compilation process significantly harder.}
The crux of Step~3 is to show that UCQ entailment \emph{is} finitely controllable as long as the set of used role names consists of a single transitive role name and queries are acyclic. Going beyond existing results \cite{BienvenuELOS10}, we show that the size of the countermodel can be bounded independently from the query. 

The structure of the paper is as follows. We give the necessary preliminaries in Section~\ref{sc:preliminaries}. Section~\ref{sec:lower} contains the proof of the \TwoExpTime-lower bound. Section~\ref{sec:acyclic} focuses on the case of a single transitive role and acyclic queries, needed for Step 3. In Section~\ref{sec:ptqs} we revise the notion of PTQs. In Section~\ref{sec:upper} we implement the three steps and  obtain both upper bounds. We conclude in Section~\ref{sc:conclusions}.

Full proofs for all statements are deferred to the appendix.

\section{Preliminaries}
\label{sc:preliminaries}

\paragraph{TBoxes, ABoxes, and Knowledge Bases.} 

We fix countably infinite sets $\NI$ of \emph{individual names}, $\NC$ of 
\emph{concept names}, and $\NR$ of \emph{role names}, partitioned into non-transitive role names $\NRn$ and transitive role names $\NRt$.
\emph{Concepts} $C$ of the description logic $\Smc$ are defined by the grammar:
\[C: :=A\mid \neg C\mid C_1\sqcup C_2\mid\exists r.~C\,.\]
A \emph{concept inclusion (CI)} is an expression of the form
$C\sqsubseteq D$ for concepts $C,D$. A \emph{TBox} is a
finite set of concept inclusions.
An \emph{ABox} is a finite set of concept assertions $A(a)$ and role assertions $r(a,b)$ for $A\in\NC$, $r\in \NR$, and $a,b\in\NI$. A \emph{knowledge base (KB)} is a pair $\Kmc=\pair{\tbox}{\abox}$ consisting of a TBox $\tbox$ and an ABox $\abox$.
We write $\NC(\tbox)$, $\NI(\abox)$, etc. for the finite sets of concept names, individual names, etc. that occur in a particular TBox $\tbox$ or ABox $\abox$.

\paragraph{Interpretations.} The semantics of concepts, TBoxes, and ABoxes are defined as usual based on \emph{interpretations}. An \emph{interpretation} is a tuple $\interp=\pair{\Delta^\Imc}{\cdot^\Imc}$ where $\Delta^\Imc$ is the \emph{domain} of $\Imc$; and $\cdot^\Imc$ assigns a subset $A^\Imc\subseteq \Delta^\Imc$ of the domain to every concept name $A\in \NC$, a binary relation $r^\Imc\subseteq \Delta^\Imc\times \Delta^\Imc$ to every role name $r\in \NR$, and an element $a^\Imc\in\Delta^\Imc$ to every individual name $a\in\NI$~\cite{DLBook}. The interpretation of complex concepts is standard: 
\begin{align*}
  (\neg C)^\Imc &= \Delta^\Imc \setminus C^\Imc, \quad
  (C\sqcup D)^\Imc  = C^\Imc \cup D^\Imc, \\
  (\exists r.~C) ^\Imc & = \{ b\in \Delta^\Imc \mid \pair{b}{c}\in r^\interp\text{ for some $c\in C^\Imc$}\}.
\end{align*}
An interpretation \Imc is a \emph{model} of a TBox $\tbox$, written $\Imc\models\tbox$, if $t^\Imc$ is a transitive relation, for all transitive role names $t\in\NR^t$, and $C^\Imc\subseteq D^\Imc$ for every concept inclusion~$C\sqsubseteq D$ in $\tbox$. For the semantics of ABoxes, we adopt the \emph{standard name assumption}, that is, $a^\mathcal I=a$ for all $a\in\NI(\Amc)$.
Then, \Imc is a \emph{model} of an ABox $\abox$, written $\Imc\vDash\abox$, if~$a\in A^\interp$ for every assertion $A(a)\in\abox$ and $\pair{a}{b}\in r^{\interp}$ for every assertion $r(a,b)\in\abox$. Finally, \Imc is a model of a KB $\Kmc=\pair{\tbox}{\abox}$ if $\Imc\models\Tmc$ and $\Imc\models\Amc$.


For a domain element $d\in \Delta^\interp$, the \emph{type of $d$ in $\interp$} is defined as $\tp(\interp, d) = \{A\in \NC \mid d\in A^\interp\}$.

The \emph{transitive closure} of an interpretation $\Imc$ is the interpretation $\Imc^+$ that coincides with \Imc except that, for every $t\in\NRt$, $t^{\Imc^+}$ is the transitive closure of $t^\Imc$.

%

A \emph{tree} is a directed acyclic graph in which exactly one node (the \emph{root}) has no incoming edges, and each other node has exactly one incoming edge (originating in the \emph{parent} of the node). With an interpretation $\Imc$ we associate a  directed multigraph $G_\Imc$ in which nodes are the domain elements and edges are obtained by taking the disjoint union of the interpretations of all role names. We call $\Imc$ \emph{tree-shaped} if $G_\Imc$ is a  tree (in particular, it has no parallel edges). 
We call $\Imc$ a \emph{transitive-tree} interpretation if it is the transitive closure of a tree-shaped interpretation $\Imc_0$. The root of $\Imc$ is the root of $\Imc_0$.

\paragraph{Conjunctive Queries.}
Let $\NV$ be a countably infinite set of \emph{variables}.
A \emph{conjunctive query (CQ)} is an expression of the form $q(\bar x)$ where $q$ is a finite set of \emph{atoms} 
of the form
$A(x)$ or $r(x,y)$ where $x,y\in \NV$, $A\in\NC$, and $r\in\NR$, and $\bar x$ is a tuple of variables occurring in  the atoms of $q$. We call $\bar x$ the \emph{answer variables} of $q(\bar x)$. We write $\var(q)$ for the set of all variables occurring in $q$.  
A \emph{union of conjunctive queries (UCQ)} $Q(\bar x)$ is a finite set of CQs with the same answer variables $\bar x$, which we call the answer variables of $Q$. 
A (U)CQ is \emph{Boolean} if its tuple of answer variables is empty, and \emph{unary} if it is a singleton.  We identify a Boolean CQ $q$ with its set of atoms.


A \emph{match} of a CQ $q(\bar x)$ in an interpretation $\interp$ is a function $\delta\colon\var(q)\rightarrow \Delta^{\interp}$ such that $\delta(x)\in A^{\interp}$ for each $A(x)\in q$, and  $\pair{\delta(x)}{\delta(y)}\in r^{\interp}$ for each $r(x,y)\in q$. For a tuple $\bar d$ of domain elements from $\Delta^\Imc$, we write $\pair{\Imc}{\bar d}\models q(\bar x)$ if there is a match $\delta$ of $q$ in \Imc with $\delta(\bar x)=\bar d$.
For UCQs, we write $\pair{\Imc}{\bar d}\models Q(\bar x)$ if $\pair{\Imc}{\bar d}\models q(\bar x)$ for some $q(\bar x)\in Q(\bar x)$. In the Boolean case, we write $\Imc\models q$ if there is a match of $q$ in $\Imc$, and $\Imc\models Q$ if $\Imc\models q$ for some $q\in Q$. 


To any CQ $q$ one can associate a directed multigraph \( G_q \), where nodes represent variables and edges are formed by binary atoms. If atoms share the same pair of variables, this creates parallel edges in \( G_q \). A query \( q(\bar{x}) \) is called \emph{acyclic} or \emph{connected} if \( G_q \) is acyclic or connected, respectively. A query is said to be \emph{rooted} if it is connected and not Boolean. These definitions extend to UCQs: a UCQ \( Q \) is \emph{acyclic}, \emph{connected}, or \emph{rooted} if each CQ in \( Q \) is.
We call \emph{tree query} (TQ) any unary CQ $q(x)$ such that $G_q$ is a directed tree 
(in particular, it has no parallel edges), $x$ being its root. A \emph{union of tree queries} (UTQ) is a UCQ that contains only TQs. 

For a CQ $q(\bar x)$, a variable $z \in \var(q)$ is \emph{initial in $q$} if $z$ has no incoming edges in  $G_q$. For instance, if $q(x)$ is a TQ, then $x$ is initial in $q(x)$.
 
\paragraph{Query Entailment.} Let $\Kmc=\pair{\Tmc}{\Amc}$ be a KB, $Q(\bar x)$ be a UCQ and  $\bar a$ be a tuple of individuals from $\NI(\Amc)$. We say that 
$\Kmc$ \emph{entails} $Q(\bar a)$, written $\Kmc\models Q(\bar a)$, if $\pair{\interp}{\bar a}\vDash Q$ for every model $\Imc$ of $\Kmc$. We study the reasoning problems of UCQ entailment and rooted UCQ entailment. \emph{UCQ entailment} asks, given a KB $\Kmc$ and a Boolean UCQ $Q$, whether $\Kmc\models Q$; and \emph{rooted UCQ entailment}
asks, given a KB $\Kmc$, a rooted UCQ $Q(\bar x)$, and a tuple $\bar a$ from $\NI(\Amc)$, whether $\Kmc\models Q(\bar a)$. We also consider the variant \emph{over a single transitive role $t$} which means that the only role that occurs in $\Tmc$ and $Q$ is $t$.
More general query answering problems, for example, for UCQs with constants, can be reduced to the above using standard methods~\cite{GlimmLHS08}.

Throughout the paper, we assume that TBoxes $\Tmc$ be in \emph{normal form} which means that each concept inclusion in \Tmc has one of the following shapes: 
\[\textstyle\bigsqcap_i A_i\sqsubseteq\bigsqcup_j B_j,\quad  A\sqsubseteq \exists r.B,\quad  A\sqsubseteq\forall r.B,\] 
where $A,A_i\in\NC\cup\{\top\}$, $B\in \NC$, $B_j\in \NC\cup\{\bot\}$; $\bot,\top,\,\sqcap,$ and $\forall$ are part of the syntax, with standard semantics. 
It is routine to show that this is without loss of generality for the considered entailment problems.

When stating complexity bounds, we write $\|\tbox\|$ and $\|Q\|$ for the size of $\Tmc$ and $Q$, respectively, represented as a word over a suitable alphabet. 



\section{Two Transitive Roles}\label{sec:lower}

In this section, we show that  UCQ entailment in $\Smc$ is  \TwoExpTime complete when the query involves at least two transitive roles.  The \TwoExpTime-upper bound follows from several previous works, e.g.,~\cite{GlimmLHS08,CalvaneseEO14,DBLP:journals/ai/Gutierrez-Basulto23, GottlobPT13}. We focus on the hardness of the problem.

\begin{restatable}{theorem}{thmlower} \label{thm:2expComp}
UCQ Entailment in $\Smc$ is \TwoExpTime-complete. It is \TwoExpTime-hard already for CQs and ABoxes of shape $\{A(a)\}$, as long as at least two transitive roles are available.
\end{restatable}

Our proof closely follows the \TwoExpTime-hardness proof for CQ entailment in $\SH$ provided in~\cite{EiterLOS09}. Since a simple reduction from this problem seems impossible, we give a direct argument. Here we sketch \TwoExpTime-hardness for \emph{unions} of CQs, which is significantly easier than for CQs. 

We reduce from the word problem for exponential-space alternating Turing machines.  We encode runs of such machines as interpretations of the form shown in Figure~\ref{fig:config-trees}, where edges labeled by $\alpha$ represent paths of length 2, consisting of a $t_1$-edge followed by a $t_2$-edge for $t_1, t_2 \in \NRt$. 
%
%
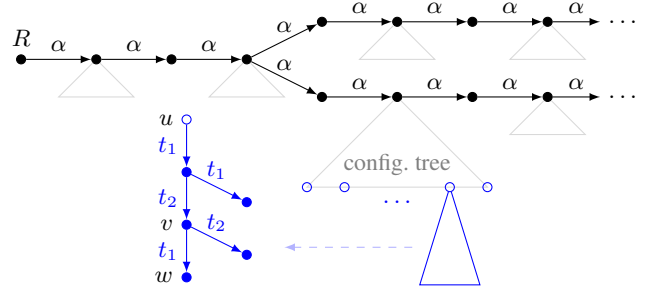
\begin{figure}
\begin{tikzpicture}[>=latex,thin,point/.style={circle,draw=black,fill,minimum size=1.2mm,inner sep=0pt},scale=1.0, bend angle=45]
    
  \footnotesize

  \node[point, label=above:{$R$}] (xr) at (-1,-0.5) {};
  \node[point] (x0) at (0,-0.5) {};
  \node[point] (y0) at (1,-0.5) {};
  \node[point] (x1) at (2,-0.5) {};
  \node[point] (y1) at (3,0) {};
  \node[point] (x2) at (4,0) {};
  \node[point] (y2) at (5,0) {};
  \node[point] (x3) at (6,0) {};
  \node (x4) at (7,0) {$\ldots$};
  \node[point] (x2') at (4,-1) {};
  \node[point] (y2') at (5,-1) {};
  \node[point] (x3') at (6,-1) {};
  \node (x4') at (7,-1) {$\ldots$};  
  \node[point] (y1') at (3,-1) {};
  
  \draw[->] (
  xr)-> node[above]{$\alpha$} (x0);
  \draw[->] (x0)-> node[above]{$\alpha$} (y0);
  \draw[->] (y0)-> node[above]{$\alpha$} (x1);
 \draw[->] (x1)-> node[above]{$\alpha$} (y1);
 \draw[->] (y1)-> node[above]{$\alpha$} (x2);
 \draw[->] (x2)-> node[above]{$\alpha$} (y2);
 \draw[->] (y2)-> node[above]{$\alpha$} (x3);
 \draw[->] (x3)-> node[above]{$\alpha$} (x4);
 \draw[->] (x2')-> node[above]{$\alpha$} (y2');
 \draw[->] (y2')-> node[above]{$\alpha$} (x3');
 \draw[->] (x3')-> node[above]{$\alpha$} (x4');
 \draw[->] (x1)-> node[above]{$\alpha$} (y1');
 \draw[->] (y1')-> node[above]{$\alpha$} (x2');    

 \draw[color=gray!30] (x2') -- (5.2, -2.2) -- (2.8, -2.2) -- (x2');
 \draw[color=gray!30] (x2) -- (4.5, -.5) -- (3.5, -.5) -- (x2);
 \draw[color=gray!30] (x0) -- (-0.5, -1) -- (0.5, -1) -- (x0);
 \draw[color=gray!30] (x1) -- (1.5, -1) -- (2.5, -1) -- (x1);
 \draw[color=gray!30] (x3) -- (5.5, -.5) -- (6.5, -.5) -- (x3);
 \draw[color=gray!30] (x3') -- (5.5, -1.5) -- (6.5, -1.5) -- (x3');
 
 \node at (4,-1.9) {\textcolor{gray}{config.~tree}};

 \node[circle,draw=blue,minimum size=1.2mm,inner sep=0pt,label=left:{$u$}] (np) at (1.2, -1.3) {};
 
 \node[point,color=blue] (np1) at (1.2, -2.0) {};
 \node[point,color=blue,label=left:{$v$}] (np2) at (1.2, -2.7) {};
 \node[point,color=blue,label=left:{$w$}] (np3) at (1.2, -3.4) {};

 \node[point,color=blue] (np2l) at (2, -3.1) {};
 \node[point,color=blue] (np1l) at (2, -2.4) {};

 \draw[->,blue] (np) -- node[left]{$t_1$}(np1); 
 \draw[->,blue] (np1) -- node[left]{$t_2$}(np2); 
 \draw[->,blue] (np2) -- node[left]{$t_1$}(np3);

 \draw[->,blue] (np1) -- node[above]{$t_1$} (np1l);
 \draw[->,blue] (np2) -- node[above]{$t_2$} (np2l);

  \node[circle,draw=blue,minimum
    size=1.2mm,inner sep=0pt] at (5.2,-2.2) {};
  \node[circle,draw=blue,minimum size=1.2mm,inner sep=0pt] (g) at (4.7,-2.2) {};
  \node at (4,-2.4) {$\textcolor{blue}{\ldots}$};
  \node[circle,draw=blue,minimum
    size=1.2mm,inner sep=0pt] at (3.3,-2.2) {};
  \node[circle,draw=blue,minimum
    size=1.2mm,inner sep=0pt] at (2.8,-2.2) {};

   \draw[color=blue] (g) -- (5.1, -3.5) -- (4.3, -3.5) -- (g);

   \draw[->,dashed,color=blue!30] (4.2,-3) -- (2.5,-3);

 \end{tikzpicture} 
\caption{Encoding runs of alternating Turing machines.}
\label{fig:config-trees}
\end{figure}
Each gray triangle represents a configuration: existential ones have one successor, and universal two. Each of these triangles is a full binary tree of height $n$, built from $\alpha$-edges, whose $2^n$ leaves correspond to tape cells. A tape cell is encoded using the blue gadget on the left. The content of the cell is stored in node $w$: the current using concept names $Y_\sigma$ and the previous one with $Z_\sigma$ where $\sigma$ ranges over possible cell contents. Nodes $u$ and $v$ encode the number of the cell in binary using concept names $B_1, \dots, B_n$. Each $B_i$ is  present in exactly one of the nodes $u$ and $v$: in $u$ if the $i$th bit is 0, and in $v$ if it is 1. 

With a bit of effort one can define a TBox (using additional concept names to propagate information) that ensures that the interpretation correctly encodes a run of the machine, provided that previous tape content (along with the state annotation) is correctly copied from the previous configuration. The latter is ensured by the query, which is the union of CQs $q_{\sigma,\tau}$ for $\sigma,\tau$  ranging over pairs of \emph{different} cell contents. Each $q_{\sigma,\tau}$ detects a copying error where $\sigma$ was replaced by $\tau$. It is shown in Figure~\ref{fig:detecting-errors}, with the dashed edges  representing the path on the right,
directed from $x$ to $x'$.

\begin{figure}
\begin{tikzpicture} [>=latex,thin,point/.style={circle,draw=black,fill,minimum
    size=1.2mm,inner sep=0pt},scale=1.0, bend angle=45]\footnotesize

  \node[point, label=right:{$x_1$}] (e0) at (.6,0) {};
  \node[point, label=left:$B_1$,label=right:$y_1$] (e1) at (0,-1) {};
  \node[point, label=left:$B_1$,label=right:$z_1$] (f1) at (1.2,-1.3) {};

  \node[point, label=left:{$Y_\sigma$},label=below:{$y$}] (c2) at (0,-2) {};
  \node[point, label=left:{$Z_\tau$},label=below:{$z$}] (d2) at (1.2,-2.3) {};

  \draw[->] (e0) -> node[left] {$\alpha^{n+1}$}(e1);
  \draw[->] (e0) -> node[right] {$\alpha^{n+3}$} (f1);
  \draw[->,dashed] (e1) -> (c2);
  \draw[->,dashed] (f1) -> (d2);

  \node[point, label=right:{$x_n$}] (g0) at (3.8,0) {};
  \node[point, label=right:$B_n$,label=below:$y_n$] (g1) at (3.2,-1) {};
  \node[point, label=right:$B_n$,label=below:$z_n$] (h1) at (4.4,-1.3) {};

  \node at (2,0) {$\ldots$};

  \draw[->] (g0) -> node[left] {$\alpha^{n+1}$}(g1);
  \draw[->] (g0) ->node[right] {$\alpha^{n+3}$} (h1);
  \draw[->,dashed] (g1) -> (c2);
  \draw[->,dashed] (h1) -> (d2);

  \node[point, label=left:{$x$}] (x0) at (6,0) {};
  \node[point] (y0) at (7,-1) {};
  \node[point] (x1) at (6,-0.8) {};
  \node[point] (y1) at (7,-1.8) {};
  \node[point] (x2) at (6,-1.6) {}; \node[point,label=left:{$x'$}] (y2) at (6,-2.3) {};

  \draw[->] (x0)-> node[right]{$\:t_1$} (y0);
  \draw[->] (x1)-> node[above]{$t_1\;\;$} (y0);
 \draw[->] (x1)-> node[right]{$\:t_2$} (y1);
 \draw[->] (x2)-> node[above]{$t_2\;\;$} (y1);
 \draw[->] (x2)-> node[right]{$t_1$} (y2);
    
\end{tikzpicture}
    
\caption{Detecting copying errors.}
\label{fig:detecting-errors}
\end{figure}
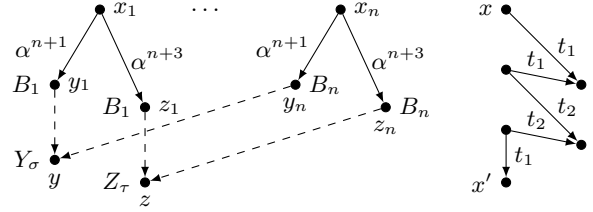

Variables $y$ and $z$ can only match in consecutive configurations, in  the $w$ nodes of some cell gadgets, say $w_\sigma$ and $w_\tau$. Then, $y_i$ can match only in  $u_\sigma$ or $v_\sigma$, and  $z_i$ only in $u_\tau$ or $v_\tau$. Moreover, owing to the rigidity of the paths $\alpha^{n+1}$ and $\alpha^{n+3}$, $y_i$ matches in $u_\sigma$ iff $z_i$ matches in $u_\tau$. Hence, $y$ and $z$ can only match in the same cell of two consecutive configurations. They do if and only if this cell was copied incorrectly.

\section{Acyclic Queries, Single Transitive Role}
\label{sec:acyclic}

We now focus on the special case of \emph{acyclic} queries over a single transitive role, which is the crux of the problem. We prove an exponential countermodel property. 

\begin{theorem}  \label{th:acyclic_transitive}
Consider an acyclic Boolean UCQ $Q$ and a TBox $\Tmc$, {both of which use a single role name $t$, which is transitive. Then, for every transitive-tree interpretation $\Imc$ over $t$} with root $d_0$ such that both $\Imc \models \Tmc$ and $\Imc \not\models Q$, there exists an interpretation $\Jmc$ satisfying the following:
\begin{itemize}

\item $\Delta^\Jmc\subseteq\Delta^\interp$ and $\tp(\Jmc, d)=\tp(\interp, d)$ for each $d\in\Delta^\Jmc$;

\item $\Jmc$ is finite and $|\Delta^{\Jmc}| \leq (|\NC(\tbox)|+1)!$;

\item $d_0\in \Delta^\Jmc$, $\Jmc \models \Tmc$, and $\Jmc \not \models Q$.

\end{itemize}
\end{theorem}



{Unexpectedly, yet crucially, the bound does not depend on the query at all, unlike in \cite{BienvenuELOS10}. Moreover, it depends only on the number of concept names mentioned in $\Tmc$, rather than all concept names occurring in $\interp$.}

The rest of the section is devoted to the proof of Theorem~\ref{th:acyclic_transitive}.
%
%
Let us fix $Q$, $\Tmc$, $\Imc$, and $d_0$ as in the statement.
While $\interp$ may well be infinite, the extracted interpretation $\Jmc$ will be finite, but not necessarily 
a transitive-tree interpretation.
For $d\in\Delta^\Imc$, define $\desc{d} = \big\{ e\in\Delta^\Imc \bigm| \pair{d}{e}\in t^\Imc\big\}$ and let $\Imc_d$ be the subinterpretation of $\Imc$ induced by $\{d\}\cup \desc{d}$.

Let $m=|Q|$. We define a function $\bar{\qbf}$ that assigns to each $d\in\Delta^\interp$ an $m$-tuple $\bar{q}_d = \langle q_{d,1}, \dots, q_{d,m}\rangle$ of acyclic Boolean CQs that are \emph{forbidden} in $\Imc_d$, in the sense that
\begin{equation} 
\Imc_d\not\models \{q_{d,i}\mid 1\leq i\leq m\}\text{ for each $d\in \Delta^\Imc$.}\label{eq:invariant}
\end{equation}
The function $\bar{\qbf}$ is defined by induction. First, we set $\bar{q}_{d_0}$ as any $m$-tuple listing all CQs from $Q$.
We then proceed top-down, maintaining the invariant. Suppose that $\bar q_d$ is already defined for some $d\in\Delta^\interp$; we shall define $\bar{q}_e$ for all \emph{direct $t$-successors} $e$ of $d$ in $\Imc$: elements $e\in\desc{d}$ such that there is no $f\in\desc{d}$ with $e\in \desc{f}$.
For each $i\leq m$, let $X_{d,i} \subseteq \var(q_{d,i})$ be the set of initial variables of $q_{d,i}$ that can be matched in $d$; that is, it contains an initial variable $x$ of $q_{d,i}$ iff $d \in A^\interp$ for each atom $A(x)$ in $q_{d,i}$.
Consider the query $q'$ obtained from $q_{d,i}$ by dropping all atoms involving a variable from $X_{d,i}$, and let us look at its connected components. If each connected component of $q'$ admitted a match in a (strict) subtree of $\Imc_d$
then by merging these matches and mapping each variable from $X_{d,i}$ to $d$, we would get a match of $q_{d,i}$ in $\Imc_d$, contradicting the assumption. 
Hence, there must be a connected component $q'_{d,i}$ of $q'$ that does not admit such a match. 
(Notice that $q'_{d,i} = q_{d,i}$ iff  $X_{d,i}=\emptyset$.) 
We let $q_{e,i} = q'_{d,i}$ for each direct $t$-successor $e$ of $d$ in $\Imc$.
Crucially, $\bar{\qbf}$ is anti-monotone: $\bar q_e\preceq \bar q_d$ for each $\pair{d}{e}\in {t^\interp}$, where $\bar q_e\preceq \bar q_d$ iff $q_{e,i}$ is a subquery of $q_{d,i}$ for each $i$. 

For our construction of $\Jmc$, we also formalize the notion of `visible concepts' of an element in the tree. For each $d\in \Delta^\interp$, we define $\rch^{\Imc}_t(d)$ as the set 
\[\{A\in \NC(\tbox)\mid e\in A^\Imc\text{ for some  $e\in \desc{d}$}\}.\] 
Because of the transitivity of $t^{\interp}$, the function $\rch^{\Imc}_t$ is also anti-monotone: $\rch^{\Imc}_t(e)\subseteq\rch^{\Imc}_t(d)$ for every $\pair{d}{e}\in t^{\interp}$.

Now that we defined the functions $\bar{\qbf}$ and $\rch^\Imc_t$, we can finally use them to construct, for each $d\in\Delta^\interp$, a finite interpretation $\Jmc_d$ with the following properties: 
\begin{itemize}

  \item $d \in \Delta^{\Jmc_d}\subseteq \Delta^\interp$, 
  
  \item $\tp(\Jmc_{d}, e)=\tp(\interp, e)$ for all $e\in\Delta^{\Jmc_{d}}$,
  
  \item  $\rch^{\Imc}_t(e)=\rch^{\Jmc_d^+\!}_t(e)$ for all $e\in\Delta^{\Jmc_{d}}$,
  
  \item $\Jmc_d^+ \not \models \{q_{d,i} \mid i\leq m\}$, and 
  
  \item $|\Delta^{\Jmc_d}|\leq (|\rch^{\Imc}_t(d)|+1)!$. 
  
\end{itemize}
The second and third condition together immediately give $\Jmc_d^+ \models \Tmc$, because $\Tmc$ is in normal form and only uses the role name $t$. The fourth states that $\Jmc_d^+$ satisfies the same invariant~\eqref{eq:invariant} as $\Imc_d$. In consequence, defining $\Jmc$ as $\Jmc_{d_0}^+$ will complete the proof of Theorem~\ref{th:acyclic_transitive}.

We construct $\Jmc_d$  for $d\in\Delta^\interp$ by induction over the set $M=\rch^\interp_t(d)\subseteq\NC(\tbox)$.
There are two cases, depending on the set
$\alpha_\rchind=\big\{ e\in \Delta^{\Imc_d} \mid \rch^{\Imc}_t(e)=\rchind\big\}$. 
The first case is conceptually simpler and serves also as the induction base.

\paragraph{1. There exists $e\in \alpha_\rchind$ such that $\alpha_\rchind\cap \desc{e}=\emptyset$.} Hence, $\rch_t^\Imc(e)=\rchind$ but $\rch_t^\Imc(f)\subsetneq \rchind$ for all $f\in \desc{b}$. Then, we can select elements $f_1,\ldots,f_\ell\in \desc{e}$, $\ell\leq|\rchind|$, such that for each $A\in\rchind$, there exists $j$ with $f_j\in A^\Imc$. 
We build $\Jmc_a$ by taking the disjoint union of interpretations $\Jmc_{f_j}$, that exist by the induction hypothesis, and adding element $d$ with unary type inherited from $\interp$ along with $t$-edges from $d$ to $f_j$ for all $j$, as shown in Figure~\ref{fig:dicho} (left). Note that in the induction base, where $\rchind=\emptyset$, $\Jmc_d$ has domain $\{d\}$ and no edges.

        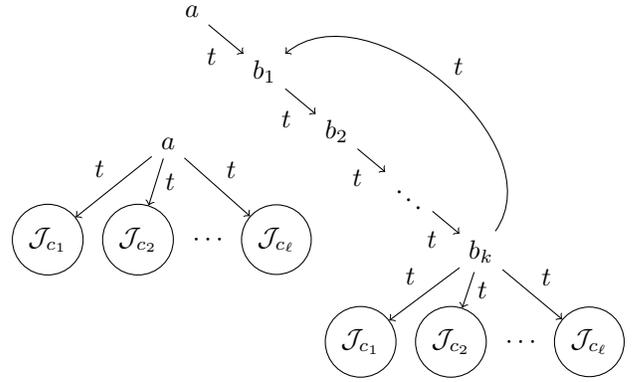
\begin{figure}
            \centering
        \begin{tikzpicture}[scale = 0.8]

\node (aa) at (-5,-2.2) {$d$};
            \node[draw, circle, minimum size = 10pt] (0) at (-7, -3.8) {$\Jmc_{f_1}$};
            \node[draw, circle, minimum size = 10pt] (1) at (-5.5, -3.8) {$\Jmc_{f_2}$};
            \node (dots) at (-4.3, -3.8) {$\cdots$};
            \node[draw, circle, minimum size = 10pt] (n) at (-3.2, -3.8) {$\Jmc_{f_{\ell}}$};
            \draw[->] (aa) -- node[midway, above left] {$t$} (0);
            \draw[->] (aa) -- node[midway, right] {$t$} (1);
            \draw[->] (aa) -- node[midway, above right] {$t$} (n);

            \node (a) at (-4.6, 0) {$d$};
            \node (b1) at (-3.4, -1) {$e_1$};
            \node (b2) at (-2.2, -2) {$e_2$};
            \node (dots) at (-1.0, -3) {$\ddots$};
            \node (bn) at (0.2, -4) {$e_{k}$};
            \draw[->] (a) -- node[midway, below left] {$t$} (b1);
            \draw[->] (b1) -- node[midway, below left] {$t$} (b2);
            \draw[->] (b2) -- node[midway, below left] {$t$} (dots);
            \draw[->] (dots) -- node[midway, below left] {$t$} (bn);
            \draw[->, out=55,in=40,looseness=1] (bn) to node[midway, above right] {$t$} (b1);
            \node[draw, circle, minimum size = 10pt] (0) at (-1.8, -5.5) {$\Jmc_{d_1}$};
            \node[draw, circle, minimum size = 10pt] (1) at (-0.3, -5.5) {$\Jmc_{d_2}$};
            \node (dots) at (0.9, -5.5) {$\cdots$};
            \node[draw, circle, minimum size = 10pt] (n) at (2, -5.5) {$\Jmc_{d_{\ell}}$};
            \draw[->] (bn) -- node[midway, above left] {$t$} (0);
            \draw[->] (bn) -- node[midway, right] {$t$} (1);
            \draw[->] (bn) -- node[midway, above right] {$t$} (n);
        \end{tikzpicture}
            \caption{Interpretation $\Jmc_d$ in the two cases.}
            \label{fig:dicho}
        \end{figure}

\paragraph{2. For all $e\in \alpha_\rchind$, $\alpha_\rchind\cap \desc{e}\neq \emptyset$.}
Then, because $\preceq$ is a well-founded partial order, there is some $e_0\in\alpha_\rchind$ and a $\preceq$-minimal $m$-tuple $\bar{q}$ such that $\bar q_e=\bar q$ for all $e\in \desc{e_0}\cap \alpha_\rchind$. Since $e_0\in \alpha_\rchind$, we can, 
as in the first case, select elements $f_1,\ldots,f_m\in \desc{b_0}$, $m\leq |\rchind|$, such that for each $A\in M$, there is some $j$ with $f_j\in A^{\Imc}$. We partition $f_1,\ldots,f_m$ in two sequences, depending on $\rch_t^\Imc$: 
\begin{itemize}

    \item $d_1\ldots,d_\ell$ contains all those $d_i$ with $\rch_t^{\Imc}(d_i)\subsetneq M$, and
    
    \item $e_1,\ldots,e_k$ contains all those $d_i$ with $\rch_t^{\Imc}(d_i)=M$.
    
\end{itemize}

To build $\Jmc_a$, we first arrange $a, e_1, e_2, \dots, e_k$ (with unary types inherited from $\interp$) into a simple path with $t$-edges. Next, we turn it into a cycle by adding a $t$-edge from $e_k$ to $e_1$. Finally, we add the disjoint union of interpretations $\Jmc_{d_1}, \dots, \Jmc_{d_\ell}$, along with a $t$-edge from $e_k$ to $d_j$ for all $j\leq \ell$. See Figure~\ref{fig:dicho} (right) for an illustration of the constructed $\Jmc_d$.

\medskip

{Verifying that $\Jmc_d$ has the desired properties is easy, except for $\Jmc_{d}^+\not\models \{q_{d,i} \mid i\leq m\}$ which subtly relies on the choice of $e_0$ and the anti-monotonicity of $\preceq$.


\section{Pseudo-Tree Queries Differently}
\label{sec:ptqs}


%


{To prepare for multi-role queries, we revise the notion of pseudo-tree queries (PTQs).} 
The key property missing in \cite{EiterLOS09}~is global undirected acyclicity, which we build into our definition. 
Intuitively, we define a PTQ as a connected CQ  whose set of binary atoms can be partitioned into disjoint connected acyclic sets of atoms over the same role name, called \emph{clusters}, that are arranged into a tree. The latter condition ensures  global undirected  acyclicity. Let us make this precise.

\begin{definition}
Let $q(\bar x)$ be a CQ. 
For a non-transitive role name $r$, an \emph{$r$-cluster} of $q(\bar x)$ is a nonempty maximal subset $C_r$ of $q$ of the form $\{r(x,y_1),r(x,y_2),\dots, r(x,y_k)\}$ with $k>0$.
For a transitive role name $t$, a \emph{$t$-cluster} of $q(\bar x)$ is a nonempty maximal connected set $C_t$ of $t$-atoms of $q(\bar x)$.
\end{definition}

The clusters of a CQ constitute a partition of the set of its binary atoms. We treat clusters as (Boolean) conjunctive queries. In particular, we speak of initial variables in clusters. For instance, Figure~\ref{fig:ptq} (left) shows an example of a query with one $t$-cluster and three $s$-clusters; cluster $C_4$ has two initial variables, $x$ and $u$; variables $y$ and $z$ are not initial in $C_4$, but they are initial in $C_2$ and $C_3$, respectively.


\begin{definition}
\label{def:clustertree} A \emph{cluster tree} for a CQ $q(\bar x)$ is a tree having for its set of nodes the set of clusters of $q(\bar x)$, and such that:
\begin{enumerate}[label=(\alph*),leftmargin=*]
\item \label{it:shared_variables} two clusters can only share variables if they are siblings or if one is a child of the other;
\item \label{it:siblings} two siblings can only share a variable if they also share it with their parent;
\item \label{it:parents} each non-root cluster $C$ shares with its parent exactly one variable, called the \emph{entry variable of $C$}, and this variable must be initial in $C$ (the root cluster has no entry variable).
\end{enumerate}
\end{definition}

Figure~\ref{fig:ptq} (middle) shows a cluster tree for the query on the left. Intuitively, a cluster tree reflects which clusters must be matched below each other, when the query is matched in {a transitive-tree interpretation}. However, this intuition breaks down as soon as the entry variable of a child cluster is initial in the parent cluster (as for $C_1$ and $C_4$ in Figure~\ref{fig:ptq}): then, the child cluster need not be matched below the parent cluster.

For similar reasons, cluster trees are not unique: the root of a cluster tree may be swapped for any of its children whose entry variable is initial in the current root. In Figure~\ref{fig:ptq}, an alternative cluster tree is obtained by making $C_4$ a child of $C_1$. Once we fix the root, the cluster tree is unique: all clusters sharing a variable with the root become its children, etc. By a \emph{root cluster in}  $q(\bar x)$ we mean any cluster that is the root of \emph{some} cluster tree for $q(\bar x)$. In Figure~\ref{fig:ptq}, $C_1$ and $C_4$ are root clusters, while $C_2$ and $C_3$ are not. 

\begin{definition}
\label{def:ptq} 
A \emph{Boolean pseudo-tree query (Boolean PTQ)} is a connected Boolean CQ $q$ such that 
\begin{enumerate}
\item \label{it:cluster_tree} there is a cluster tree for $q$;
\item \label{it:cluster_acyclicity} for each transitive $t$, every $t$-cluster of $q$
is acyclic.
\end{enumerate}
A \emph{unary pseudo-tree query (unary PTQ)} is a unary CQ $q(x)$ such that $q$ is a Boolean PTQ, and $x$ is initial and belongs to a single, root cluster in $q$.
A Boolean/unary UPTQ is a UCQ that contains only Boolean/unary PTQs. 
\end{definition}

 Coming back to the example in Figure~\ref{fig:ptq}, the query $q$ shown on the left is a Boolean PTQ and  $q(u)$ is a unary PTQ, whereas $q(x)$, $q(y)$, and $q(z)$ are not.
 \new{Moreover, any TQ whose answer variable occurs in at most one binary atom (as in Theorem~\ref{th:reduction-rooted} below) is a unary PTQ. On the other hand, the query in Figure~\ref{fig:naughty} (left) does not admit a cluster tree. Yet, it was classified as a PTQ in \cite{EiterLOS09}. }


\begin{figure}
    \centering
    \begin{tikzpicture} [>=latex,thin,point/.style={circle,draw=black,fill,minimum
    size=1.2mm,inner sep=0pt},scale=1.0, bend angle=45]\footnotesize

  \node[point, label=right:{\scriptsize $x$}] (x) at (1,0) {};
  \node[point, label=right:{\scriptsize $u$}] (x1) at (2,0) {};
  \node[point] (y) at (0,-.5) {};
  \node[point, label=below:{\scriptsize $y$}] (z) at (1,-.75) {};
  \node[point, label=right:{\scriptsize $z$}] (u) at (2,-1.5) {};
  \node[point] (y1) at (0,-1.3) {};
  \node[point] (y2) at (1,-2.1) {};
  
  \node[color=black!50] at (-.5,-.5) {$C_1$};
  \node[color=black!50] at (-.5,-1.3) {$C_2$};
  \node[color=black!50] at (.5,-2.1) {$C_3$};
  \node[color=red!50] at (2.7,-.4) {$C_4$};

  \draw[->] (x) -- node[above]{$s$} (y);
  \draw[->,color=red] (x) -- node[right]{$t$} (z);
  \draw[->,color=red] (x1) -- node[above]{$t$} (z);
  \draw[->,color=red] (x1) -- node[left]{$t$} (u);
  \draw[->,color=red] (z) -- node[above]{$t$} (u);
  \draw[->] (z) -- node[above]{$s$} (y1);
  \draw[->] (u) -- node[above]{$s$} (y2);
  
  \draw[color=black!30,rotate=30] (0.3,-.5) ellipse (.8cm and .25cm);
  
  \draw[color=black!30,rotate=30] (-.1,-1.15) ellipse (.8cm and .25cm);
  
  \draw[color=black!30,rotate=30] (.4,-2.3) ellipse (.8cm and .25cm);
  
  \draw[color=red!30,rotate=30] (1.1,-1.4) ellipse (.8cm and 1.1cm);

  \node[color=red] (C4) at (4,0) {$C_4$};
  \node[color=black] (C1) at (3.3,-1) {$C_1$};
  \node[color=black] (C2) at (4,-1) {$C_2$};
  \node[color=black] (C3) at (4.7,-1) {$C_3$};

  \draw[->] (C4) -> (C1);
  \draw[->] (C4) -> (C2);
  \draw[->] (C4) -> (C3);

  \node[point, label=left:{$A_{p_1(x)}$}] (a) at (6.3,0) {};
  \node[point] (b) at (7.3,0) {};
  \node[point,label=left:{$A_{p_2(y)}$}] (c) at (6.3,-.75) {};
  \node[point,label=left:{$A_{p_3(z)}$}] (d) at (7.3,-1.5) {};
 
  \draw[->,color=red] (a) -- node[right]{$t$} (c);
  \draw[->,color=red] (b) -- node[above]{$t$} (c);
  \draw[->,color=red] (b) -- node[left]{$t$} (d);
  \draw[->,color=red] (c) -- node[above]{$t$} (d);
 
   \end{tikzpicture}
    \caption{A Boolean PTQ over a transitive role $t$ and non-transitive role $s$, its cluster tree, and a query corresponding to its root cluster.}\label{fig:ptq}
\end{figure}
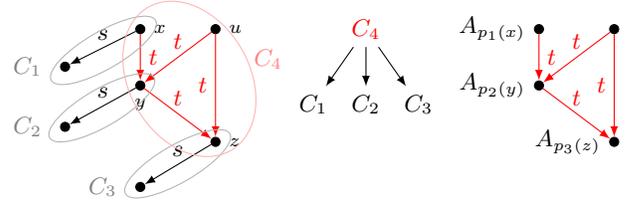

 \begin{figure}
    \centering
    \begin{tikzpicture}[>=latex,thin,point/.style={circle,draw=black,fill,minimum
    size=1.2mm,inner sep=0pt},scale=1.0, bend angle=45]\footnotesize
    \node[point, label=right:{$B$}] (x1) at (0.2,-2) {};
    \node[point] (x2) at (0.6,0) {};
    \node[point] (x3) at (0.9,-0.7) {};
    \node[point, label=above:{$A$}] (x4) at (1.9,-0.9) {};
    \node[point, label=right:{$C$}] (x5) at (1.9,-2) {};
    \node[point] (x6) at (3.2,0) {};
    \node[point] (x7) at (2.9,-0.7) {};
    \node[point, label=right:{$D$}] (x8) at (3.7,-2) {};
    
 \draw[->] (x2)-> node[left]{$s$} (x1);
 \draw[->] (x3)-> node[right]{$s$} (x1);
 \draw[->] (x2)-> node[above]{$t$} (x4);
 \draw[->] (x3)-> node[above]{$t$} (x5);
 \draw[->] (x6)-> node[above]{$t$} (x4);
 \draw[->] (x7)-> node[above]{$t$} (x5);
 \draw[->] (x6)-> node[right]{$s$} (x8);
 \draw[->] (x7)-> node[left]{$s$} (x8);
 
    \node[point] (y1) at (6,0) {};
    \node[point, label=left:{$B$}] (yb) at (5,-.8) {};
    \node[point, label=left:{$D$}] (yd) at (5.7,-.8) {};
    \node[point, label=right:{$A$}] (ya) at (6.2,-.8) {};
    \node[point, label=right:{$C$}] (yc) at (7,-.8) {};
 \draw[->] (y1)-> node[left]{$s$} (yb);
 \draw[->] (y1)-> node[left]{$s$} (yd);
 \draw[->] (y1)-> node[right]{$t$} (ya);
 \draw[->] (y1)-> node[right]{$t$} (yc);
 
    \node[point] (z1) at (6,-1.7) {};
    \node[point] (z2) at (6.5,-1.2) {};
    \node[point, label=right:{$A$}] (z3) at (7,-1.7) {};
    \node[point, label=left:{$B$}] (zb) at (5,-2.3) {};
    \node[point, label=left:{$D$}] (zd) at (6,-2.3) {};
    \node[point, label=right:{$C$}] (zc) at (7,-2.3) {};
 \draw[->] (z2)-> node[above]{$s$} (z1);
 \draw[->] (z2)-> node[above]{$t$} (z3);
 \draw[->] (z1)-> node[above]{$s$} (zb);
 \draw[->] (z1)-> node[right]{$s$} (zd);
 \draw[->] (z1)-> node[above]{$t$} (zc);
\end{tikzpicture}
\caption{A naughty query using transitive roles $s$, $t$.}
\label{fig:naughty}
\end{figure}
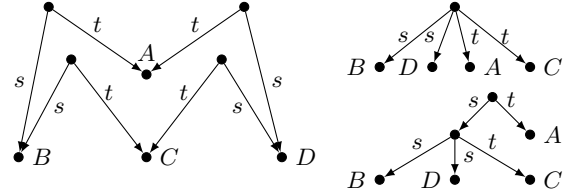

The \emph{raison d'\^etre} of PTQs is to semantically capture CQs over {transitive-tree} interpretations, as stated in Lemma~\ref{lem:to_ptq} below. Crucially, this is possible only if at most one transitive role is allowed. 

\begin{restatable}{lemma}{toPTQ}
    \label{lem:to_ptq}  Let $\mathbf T$ be the class of {transitive-tree} interpretations.
    The following can be done in polynomial time for CQs using at most one transitive role: 
    \begin{itemize}
        \item Given a connected Boolean CQ $q$, decide if ${\Imc}\models q$ for some~$\Imc\in\mathbf T$, and if so, output a Boolean PTQ $\hat{q}$ such that  ${\Imc}\models q$ iff ${\Imc}\models\hat q$ for all $\Imc\in\mathbf T$.
        
        \item Given a connected unary CQ $q(x)$, decide if $\pair{{\Imc}}{d}\models q(x)$ for some~$\Imc\in\mathbf T$ with root $d$, and if so, output unary PTQs $\hat{q}_1(x), \dots , \hat{q}_k(x)$ 
        such that for all $\Imc\in\mathbf T$ with root $d$, $\pair{{\Imc}}{d}\models q(x)$ iff $\pair{{\Imc}}{d}\models\hat q_i(x)$ for all $i$.
    \end{itemize}
\end{restatable}

The existence of $\hat q$ and \red{$\hat q_1(x), \dots, \hat q_k(x)$} in Lemma~\ref{lem:to_ptq} relies on the input query using at most one transitive role. For instance, the query $q$ in Figure~\ref{fig:naughty} (left) admits a match in 
{a transitive-tree interpretation}, but one can check that it is not captured by a single PTQ. {Intuitively, this is because $q$ is entailed by both queries in Figure~\ref{fig:naughty} (right), despite their radically different structure.}
\red{In the unary case, we need multiple (but polynomially many) PTQs, because the answer variable cannot belong to multiple clusters. For instance, for $q$ in Figure~\ref{fig:ptq}, $q(x)$ requires two unary PTQs: one consists of cluster $C_1$ and the other consists of clusters $C_2$, $C_3$, and $C_4$.} 

{We close the section by defining subPTQs, which are to PTQs what subtrees are to TQs.} A \emph{subPTQ of a Boolean PTQ} $q$ is a unary subquery of $q$ induced by a proper subtree of a cluster tree for $q$. It consists of all binary atoms in the subtree along with all unary atoms of $q$ over variables used in the subtree, and its answer variable is the entry variable of the root of the subtree. 
\red{A \emph{subPTQ of a unary PTQ $q(x)$} is defined analogously, except that we are limited to the (unique) cluster tree of $q(x)$ such that $x$ belongs to the root cluster.} 
For instance, the query $q$ in Figure~\ref{fig:ptq} has 4 subPTQs. Three are induced by the subtrees of the cluster tree in the figure, rooted at $C_1$, $C_2$, and $C_3$, and their answer variables are $x$, $y$, and $z$, respectively. 
The last one is induced by the subtree rooted at $C_4$ of the alternative cluster tree whose root is $C_1$, and its answer variable is $x$. \red{Of those 4 queries, only the first 3 are subPTQs of the unary PTQ $q(u)$.} 





\section{Upper Bounds}
\label{sec:upper}

{Our proof of \TwoExpTime-hardness of CQ entailment 
in $\Smc$ crucially relies on the availability of (a) two transitive roles and (b) Boolean CQs. We now show that entailment becomes easier if we forbid either of these. In this sense our hardness result is optimal. }


\begin{restatable}{theorem}{coNExp}
\label{th:coNExp}
The query entailment problem in $\Smc$ is $\coNExpTime$-complete for rooted UCQs and for UCQs that use at most one transitive role name. 
\end{restatable}



The \coNExpTime-hardness for UCQs using at most one transitive role was established by~\cite{EiterLOS09}. While the proof is formulated using Boolean CQs, it does not rely on this and can be adjusted easily to the case of rooted CQs. 

We focus on the upper bounds. The proofs have a common {three-step structure described in the introduction} and ultimately establish a \emph{small witness property} based on Theorem~\ref{th:acyclic_transitive}.
{We implement Steps~1--3 for both upper bounds in parallel in Sections \ref{sb:firststep}--\ref{sb:expo_sized_tiles} and we put them together} in Section~\ref{sec:wrappingup}, where the 
proof of Theorem~\ref{th:coNExp} is finalized.

\subsection{Eliminating ABoxes and Simplifying Queries}
\label{sb:firststep}

The first step eliminates the ABox and simplifies the queries. For the rooted case, this is routine. We reduce our main entailment problem to a variant  where the input consists of a TBox $\Tmc$, a set $\tau \subseteq \NC(\Tmc)$, and a unary UCQ $Q^1$: the task is to decide whether for each model $\interp$ of $\Tmc$ and every $d \in\Delta^{\interp}$ with $\tp(\interp, d) \cap \NC(\Tmc) = \tau$, we have $\pair{\interp}{a} \models Q^1$. We write
this condition as $\pair{\Tmc}{\tau} \models Q^1$. The reduction is provided by the following theorem. Note that this is a non-deterministic reduction, similar to the one introduced in \cite{nondetreductions}.

\begin{restatable}{theorem}{reductionrooted}
\label{th:reduction-rooted}
There is a \NExpTime algorithm that, given a KB $\pair{\Tmc}{\Amc}$, a rooted UCQ $Q(\bar x)$, and a tuple $\bar a$ of individuals from $\NI(\Amc)$, computes for each $a\in\NI(\Amc)$ a set $\tau_a \subseteq \NC(\Tmc)$ and a UTQ $Q^1_a$ such that
\begin{itemize}
\item $\pair{\Tmc}{\Amc} \not\models Q(\bar a)$ iff there is a run of the algorithm such that $\pair{\Tmc}{\tau_a}\not\models Q_a^1$ for all $a\in\NI(\Amc)$;
\item for each run of the algorithm and each $a\in\NI(\Amc)$, the size of each TQ in $Q_a^1$ is linear in $\|Q\|$ \red{and its answer variable occurs in at most one binary atom.}
\end{itemize}
\end{restatable}

{Next, we handle the single transitive role case.} We obtain stronger size guarantees (needed later) at the cost of relaxing tree queries (TQs) to pseudo-tree queries (PTQs). Similarly to the rooted case, we reduce entailment of UCQs using a single transitive role to a variant of entailment for UPTQs. This time, given a TBox $\Tmc$, a set $\tau \subseteq \NC(\Tmc)$, a Boolean UPTQ $Q^0$, and a unary UPTQ $Q^1$, one has to decide whether for each model $\interp$ of $\Tmc$ and each $d \in\Delta^{\interp}$  with $\tp(\interp, d) \cap \NC(\Tmc)  = \tau$, we have $\Imc\models Q^0$ or $\pair{\interp}{d} \models Q^1$. We write the condition to be decided as $\pair{\Tmc}{\tau} \models Q^0\vee Q^1$. The following theorem provides the reduction; {it relies on Lemma~\ref{lem:to_ptq} to transform CQs into PTQs equivalent over transitive-tree interpretations.

\begin{restatable}{theorem}{reductionone}
\label{th:reduction-single}
There is a \NExpTime algorithm that, given a KB $\pair{\Tmc}{\Amc}$ and a Boolean UCQ $Q$ using at most one transitive role, computes for each $a\in\NI(\Amc)$ a set $\tau_a \subseteq \NC(\Tmc)$, a Boolean UPTQ $Q^0_a$, and a unary UPTQ $Q^1_a$ such that 
\begin{itemize}
\item $\pair{\Tmc}{\Amc} \not\models Q$ iff there is a run of the algorithm such that $\pair{\Tmc}{\tau_a}\not\models Q_a^0\lor Q_a^1$ for all $a\in\NI(\Amc)$;
\item for each run of the algorithm and each $a\in\NI(\Amc)$, $\|Q_a^0\|$ is linear in $\|Q\|$, each PTQ in $Q_a^1$ is linear in $\|Q\|$, and the total number of 
subPTQs of PTQs from $Q_a^1$ is polynomial in $\|Q\|$.
\end{itemize}
\end{restatable}

Compared to Theorem~\ref{th:reduction-rooted}, Theorem~\ref{th:reduction-single} yields queries from a broader class (UPTQs, rather than UTQs), but offers stronger size guarantees: the total number of 
subPTQs of PTQs in $Q_a^1$ is polynomial. This is a substitute for a polynomial bound on $\|Q_a^1\|$, which cannot be ensured. Importantly, the variant of entailment used in Theorem~\ref{th:reduction-rooted} is a special case of the one in Theorem~\ref{th:reduction-single}: $Q_a^0 = \emptyset$ and UTQs \red{with answer variables used in at most one binary atom} instead of UPTQs. {This enables us to treat the two cases in parallel in what follows.}

\subsection{From Entailment to Existence of Mosaics}
\label{sb:mosaics}

Our goal now is to reduce the variant of UPTQ entailment introduced in Section~\ref{sb:firststep} to the special case where only one role is used in the input. Our reduction relies on \emph{tiles}, which are interpretations using only one role name, consistent with the TBox. We show that countermodels can be represented using \emph{mosaics}, which are collections of such tiles.


Below,  $\tbox_{r}$ is the restriction of $\Tmc$ to CIs that do not mention any role names other than $r$, and $\pair{\interp}{d_0}\vDash\tbox_r$ means that $d_0\in C^\Imc$ implies $d_0\in D^\Imc$ for each CI $C \sqsubseteq D$ from $\tbox_r$.

\begin{definition}
\label{def:tile}
A \emph{tile} for a TBox $\tbox$ is a triple $\langle \interp, d_0, r\rangle$ where $\interp$ is an interpretation, $d_0\in\Delta^\interp$, and $r \in \NR$, such that $s^\Imc=\emptyset$ for all $s\in\NR-\{r\}$ and 
\begin{enumerate}
\item\label{it:tbox_trans} if $r$ is transitive, then  $\interp\vDash\tbox_r$ (and, in particular, $\interp^+=\interp$);
    \item\label{it:tbox_non_trans} if $r$ is non-transitive, then  $\pair{\interp}{d_0}\vDash\tbox_r$.
\end{enumerate}
\end{definition}

Tiles are meant to be assembled to form a countermodel to $\pair{\tbox}{\tau}\vDash Q^0\vee Q^1$. To facilitate this,  we introduce a family of fresh auxiliary concept names of the form $A_{p(x)}$, where~$p(x)$ is a subPTQ of a PTQ from $Q^0$ or 
from $Q^1$. {Intuitively,  $A_{p(x)}$ will be used to propagate the constraint that $p(x)$ must not be matched. To prevent the entire CQ from being satisfied, we recursively partition it into clusters (as described in the next paragraph) and  ensure that each tile violates suitable  clusters, augmented with atoms of the form $A_{p(x)}(x)$.}
Crucially, if $Q^0$ and $Q^1$ are  $Q_a^0$ and $Q_a^1$ from Theorem~\ref{th:reduction-single} for some $Q$, then the number of auxiliary concept names  is polynomial in $\|Q\|$. 



{Let $q$ be Boolean PTQ from $Q^0$ or the Boolean PTQ underlying a subPTQ $q(x)$ of a PTQ from $Q^0$ or from $Q^1$.}
Consider a root cluster $C$ of $q$, and the corresponding cluster tree for $q$ with $C$ in the root. 
We define $q_{C}$ as the Boolean PTQ obtained by replacing each direct subtree of the cluster tree with a single unary atom using a suitable auxiliary concept name. (Recall that a direct subtree is one rooted at a child of the whole tree's root.) More precisely, $q_C$ contains all binary atoms from $C$ along with all unary atoms of $q$ over variables used in $C$, and for each child $C'$ of $C$, sharing a variable $x$ with $C$, $q_C$ contains the atom $A_{p(x)}(x)$ where $p(x)$ is the subPTQ corresponding to the cluster subtree rooted at $C'$. 

For instance, if $q$ is the query in Figure~\ref{fig:ptq} (left), then  $q_{C_4}(x)$ is the query shown on the right. Concept names $A_{p_1(x)}$, $A_{p_2(y)}$, and  $A_{p_3(z)}$ used in $q_{C_4}(x)$ replace the sub\-PTQs of $q$ induced by $C_1$, $C_2$, and $C_3$, respectively.

We now define a mosaic as a compatible collection of tiles that correctly propagates information about subPTQs. 

\begin{definition}
    \label{def:mosaic}
    Consider a TBox $\tbox$, a type $\tau$, a Boolean UPTQ $Q^0$, and a unary UPTQ $Q^1$. A \emph{mosaic for $\tbox$ and $\tau$, and against $Q^0$ and $Q^1$}, is a set $\Mos$ of tiles for $\tbox$ such that:
    \begin{enumerate}
    
        \item 
        \label{it:forbidding_boolean} 
        for each $\langle\interp, d_0, r\rangle\in\Mos$, $q\in Q^0$, and root $r$-cluster $C$ of $q$, we have $\interp\nvDash q_{C}$;
        
        \item 
        \label{it:fresh_concepts} 
        for each $\langle\interp, d_0, r\rangle\in\Mos$, $d\in\Delta^\Imc$, auxiliary $A_{p(x)}$, and root $r$-cluster $C$ of $p$ that contains $x$, if $d\notin A_{p(x)}^\interp$ then $\pair{\interp}{d}\nvDash p_{C} (x)$;
        
        \item 
        \label{it:initial} 
        there is a family of tiles \mbox{$\langle\interp_r, d_{r}, r\rangle \in \Mos$} with \mbox{$\tp(\interp_r, d_{r}) \cap \NC(\tbox) = \tau$}, for $r$ ranging over $\NR(\tbox)$, such that for each $q(x)\in Q^1$,  $\pair{\interp_r}{d_{r}}\nvDash q_{C}(x)$ for some $r\in\NR(\tbox)$ and root $r$-cluster $C$ of $q$ that contains $x$;
        
        \item 
        \label{it:types} 
        for each $\langle\interp, d_0, r\rangle\in\Mos$,  $d\in\Delta^\interp$, and $s\in\NR(\tbox)$ such that 
        either $s \neq r$ or both $r$ is non-transitive and $d\neq d_0$, 
        there is $\langle\Jmc, e_0, s\rangle\in\Mos$ such that $\tp(\interp, d)=\tp(\Jmc, e_0)$.
    \end{enumerate}
\end{definition}
As promised, our variant of the entailment problem for UPTQs reduces to the existence of mosaics.  
\begin{restatable}{theorem}{reductionbasic}
\label{th:reductionbasic}
    For any TBox $\tbox$, $\tau \subseteq \NC(\Tmc)$, Boolean UPTQ $Q^0$, and unary UPTQ $Q^1$, $\pair{\tbox}{\tau}\nvDash Q^0\vee Q^1$ iff there is a mosaic $\Mos$ for $\tbox$ and $\tau$, and against $Q^0$ and $Q^1$.
\end{restatable}

The proof of Theorem~\ref{th:reductionbasic} is technical yet relatively standard. It remains to see that the existence of suitable mosaics can be decided in \NExpTime. Towards this end, we prove in the following subsection that tiles and mosaics of singly exponential size are sufficient.


\subsection{Bounding the Sizes of Tiles and Mosaics}
\label{sb:expo_sized_tiles}

We know from Theorem~\ref{th:reductionbasic} that the entailment problem reduces to the existence of a mosaic. Now, using Theorem~\ref{th:acyclic_transitive} in Section~\ref{sec:acyclic}, we can show that we can assume each tile of the mosaic to be of exponential size.

\begin{restatable}{lemma}{corboundedsize}
\label{cor:mosaic_exp_tile}
    Consider a TBox $\tbox$, $\tau \subseteq \NC(\Tmc)$, a Boolean UPTQ $Q^0$, and a unary UPTQ $Q^1$. For every  mosaic $\Mos$ for $\tbox$ and $\tau$, and against $Q^0$ and $Q^1$, there is a mosaic $\Mos'$ for $\tbox$ and $\tau$ against $Q^0$ and $Q^1$ such that 
    $|\Delta^{\interp}| \leq (|\NC(\tbox)|+1)!$ for each tile $\langle\interp, a_0, r \rangle \in \Mos'$.
\end{restatable}

It remains to see that the number of tiles in a mosaic can be bounded as well. Here, the arguments diverge.

In the single transitive role case, we rely on the number of auxiliary concept names, determined by the number of subPTQs of PTQs from $Q^0$ and $Q^1$, being polynomial. As one tile per type is enough, we get the following. 

 \begin{restatable}{lemma}{lemboundedtilesingle}\label{lem:bounded-tiles-single}
%
%
If there is a mosaic for $\tbox$ and $\tau$, and against a Boolean UPTQ $Q^0$ and a unary UPTQ $Q^1$, {with tiles of size at most $M$}, then there is one 
with at most $n + n\cdot 2^{n+m}$ tiles, {all of size at most $M$}, where 
$n=\|\tbox\|$ and $m$ is the total number of subPTQs of PTQs from $Q^0$ and $Q^1$.
\end{restatable}

In the rooted case, the number of auxiliary concepts is exponential, but we observe that in a countermodel built from tiles, a rooted query can traverse only a linear number of tiles from the initial element. Because a tile of size $M$ requires at most $M\cdot \|\Tmc\|$ witnesses, a singly exponential number of tiles is sufficient to build the part of the countermodel within linear distance from the initial element. Further away, we do not care about matching the query anymore, so we only need one tile for each $\tau \subseteq\NC(\Tmc)$.

\begin{restatable}{lemma}{lemboundedtilesrooted}\label{lem:bounded-tiles-rooted}
If there is a mosaic for $\tbox$ and $\tau$ against $Q^0 = \emptyset$ and unary UTQ $Q^1$ with tiles of size at most $M$, then there is one with at most $n\cdot ((Mn)^{m+1}+2^n)$ tiles, {all of size at most $M$}, where $n=\|\tbox\|$ and $m$ is the maximal number of variables of a TQ in $Q^1$.
\end{restatable}

\subsection{Wrapping Up}
\label{sec:wrappingup}

Combining Theorems~\ref{th:reduction-rooted} and~\ref{th:reduction-single} (Step 1), Theorem~\ref{th:reductionbasic} (Step 2), and Lemmas~\ref{cor:mosaic_exp_tile}--\ref{lem:bounded-tiles-rooted} (Step 3), we obtain the following. 

\begin{corollary}
\label{cor:nearly_there}
There is a \NExpTime algorithm that, given a KB $\pair{\Tmc}{\Amc}$ and rooted UCQ $Q$ or a Boolean UCQ $Q$ using at most one transitive role, computes for each $a\in\NI(\Amc)$ a set $\tau_a \subseteq \NC(\Tmc)$, a Boolean PTQ $Q^0_a$ and a unary UPTQ $Q^1_a$, both containing PTQs of linear size only, such that $\pair{\Tmc}{\Amc} \not\models Q$ iff there is a run of the algorithm such that for all $a\in\NI(\Amc)$ there is a mosaic $\Mos_a$ for $\Tmc$ and $\tau_a$ against $Q^0_a$ and $Q^1_a$, of size bounded exponentially in $\|\Tmc\|+\|Q\|$ using tiles of size bounded exponentially in $\|\Tmc\|$. 
\end{corollary}

{Theorem  \ref{th:coNExp} now follows} easily, because after computing $Q_a^0$ and 
$Q_a^1$ for each $a \in \NI(\Amc)$, the algorithm from Corollary~\ref{cor:nearly_there} can guess a suitable mosaic $\Mos_a$ for each $a$. Verifying that $\Mos_a$ is indeed a mosaic for $\Tmc$ and $\tau_a$ against $Q_a^0$ and 
$Q_a^1$ can be done in time exponential in $\|\Tmc\| + \|Q\|$.

\section{Conclusions}
\label{sc:conclusions}

Contrary to previous expectations, UCQ entailment in $\Smc$ turns out to be \TwoExpTime-complete, even when restricted to trivial ABoxes and CQs with two transitive roles. 
On the positive side, both entailment of rooted UCQs and entailment of UCQs using at most one transitive role are \coNExpTime-complete and thus easier. We note a curious dependence of the complexity of the problem on the number of transitive roles allowed in queries: \ExpTime for $0$, \coNExpTime for $1$, and \TwoExpTime for at least $2$. 

Our results partly apply to UCQ entailment over finite interpretations. Indeed, it is easy to check that our \TwoExpTime-hardness proof also works in the finite case. A matching upper bound follows from~\cite{GogaczIM18}. {On the other hand, it is an open question if our $\coNExpTime$ upper bounds hold in the finite case, too.}

\bibliography{aaai26}


\cleardoublepage

\appendix

\section{Appendix for Section~\ref{sec:lower}}
\label{ap:lower}

We give a full proof of Theorem~\ref{thm:2expComp}, restated here for the reader's convenience.

\thmlower*

\new{As explained in the main part of the paper, the upper bound follows from a number of existing results. To establish the lower bound, we closely follow the one for UCQ entailment in \SH proved in ~\cite[Theorem~1]{EiterLOS09}, but for the sake of completeness, we give most of the arguments here as well, even if they are the same as in the original paper. The main novelty in our proof lies in the construction of the query, where we cannot rely on role hierarchies to create transitive super-roles that allow us to jump over multiple non-transitive  edges.  While the main idea was illustrated in Section~\ref{sec:lower}, where we sketched a hardness argument for \emph{unions} of CQs, more work is needed to  show hardness for CQs. It will involve a different cell gadget and a modification of the query, both in line with the original proof from  \cite{EiterLOS09}.}

\smallskip
As in the proof sketch for UCQs, we reduce the word problem of exponentially space-bounded alternating Turing machines (ATMs), which we briefly recall next. An ATM is given as a 5-tuple $\Mmc =(Q,\Sigma,\Gamma,q_0,\delta)$ where:
\begin{itemize}

    \item  $Q$ is a finite set of states partitioned into $Q=Q_\exists\uplus Q_\forall\uplus\{q_\mathsf{rej}\}\uplus\{q_\mathsf{acc}\}$, 
    
    \item $\Sigma$ is the input alphabet, 

    \item $\Gamma$ is the tape alphabet \new{such that $\Sigma \subseteq \Gamma$}, 
    
    \item $q_0\in Q$ is the initial state, and 
    
    \item $\delta$ is the transition function that assigns to every $(q,a)\in Q\times \new{\Gamma}$ a set $\delta(q,a)$ of transitions $(q',b,\ell)\in Q\times \new{\Gamma} \times\{-1,+1\}$, expressing that upon reading $a$ in state $q$ the ATM can write $b$, update its state to $q'$, and move the head by $\ell$ tape cells. 
\end{itemize}
A \emph{configuration} is a word $w'qw''$ with $q\in Q$ and $w',w''\in \Gamma^*$. 
\emph{Successor configurations} of $w'qw''$ are defined as usual: the first position of $w''$ is the one to be rewritten when applying a transition.  We assume  without loss of generality that there are no infinite sequences of successor configurations\new{, the initial state $q_0$ cannot be reached via a transition, and  for each $q\in Q_\forall$ and $a\in\Gamma$, there are no two transitions in $\delta(q,a)$ resulting in the same state.}
We call a configuration $w'qw''$ \emph{existential} if $q \in Q_\exists$, and \emph{universal} if  $q\in Q_\forall$. 
Accepting configurations are defined recursively, as follows. Configuration  $w'qw''$ is \emph{accepting} if:
\begin{itemize}

    \item $q=q_\mathsf{acc}$, or
    
    \item $q\in Q_\exists$ and some successor configuration of $w'qw''$  is accepting, or
    
    \item $q\in Q_\forall$ and all successor configuration of $w'qw''$ are accepting.
    
\end{itemize}
An input word $w$ is accepted by \Mmc if the \emph{initial configuration} $q_0w$ is accepting. The \emph{word problem of \Mmc} is to decide whether a given word $w\in \Sigma^*$ is accepted by \Mmc. 
\new{We call $\Mmc$  \emph{$f(n)$ space-bounded} if every sequence of successor configurations that begins from $q_0w$ with $|w|=n$ contains only configurations of the form $w'qw''$ where $|w'w''| \leq f(n)$.}
It is well-known that there is a fixed $2^n$ space-bounded ATM \Mmc whose word problem is \TwoExpTime-hard~\cite{chandraAlternation1981}. 

It is helpful to view the behavior of an ATM on input $w$ in terms of its \emph{computation tree}, which is a tree of configurations whose root is labeled with the initial configuration $q_0w$, every node labeled with an existential configuration $w'qw''$ has one successor which is labeled with a successor configuration of $w'qw''$, and every node labeled with a universal configuration $w'qw''$ has a successor node for each successor configuration of $w'qw''$. Recall that we assume that there are no infinite sequences of successor configurations, so each path in a configuration tree is finite. 
An ATM then accepts its input $w$ if there is a computation tree all of whose leaves are labeled with configurations having the state $q_\mathsf{acc}$. 

We are now ready to provide the reduction. Let us fix a  $2^n$ space-bounded ATM \Mmc whose  word problem is \TwoExpTime-hard. Given an input word $w\in \Sigma^*$ of length $n$, we will show how to compute in time polynomial in $n$ a KB $\Kmc_w=\pair{\Tmc_w}{\Amc_w}$ and a (Boolean) CQ $q_w$ such that 
\[\Mmc\text{ accepts }w\quad\text{ iff }\quad \Kmc_w\not\models q_w.\]
%
Intuitively, models of $\Kmc_w$ will represent accepting computation trees of \Mmc on input $w$, up to cell-content copying errors,  
while $q_w$ will detect  copying errors. This way, every model \Imc of $\Kmc_w$ with $\Imc\not\models q_w$ will represent an accepting computation tree of $\Mmc$ on $w$.  The ABox $\Amc_w$ is very simple, consisting of a single fact $R(a)$; the main work is thus done by $\Tmc_w$. The TBox uses two transitive role names $t_1,t_2$, and throughout the proof we denote their composition with $\alpha$. That is, in an interpretation \Imc, $\alpha^\Imc=t_1^\Imc\circ t_2^\Imc$ is the set of all pairs $(d,e)$ in which $e$ is reachable from $d$ by first following a $t_1$-edge and then a $t_2$-edge.

\begin{figure}
\centering
\begin{tikzpicture}[>=latex,thin,point/.style={circle,draw=black,fill,minimum size=1.2mm,inner sep=0pt},scale=1.0, bend angle=45]
    
  \footnotesize

  \node[point, label=left:{$R$}] (xr) at (-1,0) {};
  \node[point,label=above:{$(q_0,a_0,i_0)$}] (x0) at (0,0) {};
  \node[point] (y0) at (1,0) {};
  \node[point,label=above:{$(q_1,a_1,i_1)$}] (x1) at (2,0) {};
  \node[point] (y1) at (3,0) {};
  \node[point,label=above:{$(q_2,a_2,i_2)$}] (x2) at (4,0) {};
  \node[point] (y2) at (5,0) {};
  \node[point,label=above:{$(q_3,a_3,i_3)$}] (x3) at (6,0) {};
  \node (x4) at (7,0) {$\ldots$};
  \node[point,label=above:{$(q_4,a_4,i_4)$}] (x2') at (4,-1) {};
  \node[point] (y2') at (5,-1) {};
  \node[point,label=above:{$(q_5,a_5,i_5)$}] (x3') at (6,-1) {};
  \node (x4') at (7,-1) {$\ldots$};  
  \node[point] (y1') at (3.3,-1) {};
  
  \draw[->] (xr)-> node[below]{$\alpha$} (x0);
  \draw[->] (x0)-> node[below]{$\alpha$} (y0);
  \draw[->] (y0)-> node[below]{$\alpha$} (x1);
 \draw[->] (x1)-> node[below]{$\alpha$} (y1);
 \draw[->] (y1)-> node[below]{$\alpha$} (x2);
 \draw[->] (x2)-> node[below]{$\alpha$} (y2);
 \draw[->] (y2)-> node[below]{$\alpha$} (x3);
 \draw[->] (x3)-> node[below]{$\alpha$} (x4);
 \draw[->] (x2')-> node[below]{$\alpha$} (y2');
 \draw[->] (y2')-> node[below]{$\alpha$} (x3');
 \draw[->] (x3')-> node[below]{$\alpha$} (x4');
 \draw[->] (x1)-> node[below]{$\alpha$} (y1');
 \draw[->] (y1')-> node[below]{$\alpha$} (x2');    

 \draw[color=gray!30] (x2') -- (5.2, -3) -- (2.8, -3) -- (x2');
 \draw[color=gray!30] (x2) -- (4.5, -.5) -- (3.5, -.5) -- (x2);
 \draw[color=gray!30] (x0) -- (-0.5, -.5) -- (0.5, -.5) -- (x0);
 \draw[color=gray!30] (x1) -- (1.5, -.5) -- (2.5, -.5) -- (x1);
 \draw[color=gray!30] (x3) -- (5.5, -.5) -- (6.5, -.5) -- (x3);
 \draw[color=gray!30] (x3') -- (5.5, -1.5) -- (6.5, -1.5) -- (x3');
 
 \node at (4,-2.7) {\textcolor{gray}{config.~tree}};

 \node[circle,draw=blue,minimum
    size=1.2mm,inner sep=0pt] (n) at (1, -1) {};
 \node[point,color=blue,label=left:{$u_p$}] (np) at (0.2, -1.5) {};
 \node[point,color=blue] (np1) at (0.2, -2.2) {};
 \node[point,color=blue,label=left:{$v_p$}] (np2) at (0.2, -2.9) {};
 \node[point,color=blue,label=left:{$w_p$},label=below:{$G_p$}] (np3) at (0.2, -3.6) {};
 \node[point,color=blue,label=right:{$u_h$}] (nh) at (1.8, -1.5) {};
 \node[point,color=blue] (nh1) at (1.8, -2.2) {};
 \node[point,color=blue,label=right:{$v_h$}] (nh2) at (1.8, -2.9) {};
 \node[point,color=blue,label=right:{$w_h$},label=below:{$G_h$}] (nh3) at (1.8, -3.6) {};
 \node[point,color=blue] (np2l) at (0.8, -3.5) {};
 \node[point,color=blue] (nh2l) at (1.2, -3.5) {};
 \node[point,color=blue] (np1l) at (0.8, -2.8) {};
 \node[point,color=blue] (nh1l) at (1.2, -2.8) {};

 \draw[->,blue] (n) -- node[above]{$\alpha$} (np); 
 \draw[->,blue] (np) -- node[left]{$t_1$}(np1); 
 \draw[->,blue] (np1) -- node[left]{$t_2$}(np2); 
 \draw[->,blue] (np2) -- node[left]{$t_1$}(np3); 
 
 \draw[->,blue] (n) -- node[above]{$\alpha$}(nh); 
 \draw[->,blue] (nh) -- node[right]{$t_1$}(nh1); 
 \draw[->,blue] (nh1) -- node[right]{$t_2$} (nh2); 
 \draw[->,blue] (nh2) -- node[right]{$t_1$} (nh3); 
 
 \draw[->,blue] (np1) -- node[right]{$t_1$} (np1l);
 \draw[->,blue] (nh1) -- node[left]{$t_1$} (nh1l);
 \draw[->,blue] (np2) -- node[right]{$t_2$} (np2l);
 \draw[->,blue] (nh2) -- node[left]{$t_2$} (nh2l);

  \node[circle,draw=blue,minimum
    size=1.2mm,inner sep=0pt] at (5.2,-3) {};
  \node[circle,draw=blue,minimum size=1.2mm,inner sep=0pt] (g) at (4.7,-3) {};
  \node at (4,-3.2) {$\textcolor{blue}{\ldots}$};
  \node[circle,draw=blue,minimum
    size=1.2mm,inner sep=0pt] at (3.3,-3) {};
  \node[circle,draw=blue,minimum
    size=1.2mm,inner sep=0pt] at (2.8,-3) {};

   \draw[color=blue] (g) -- (5.1, -4) -- (4.3, -4) -- (g);

   \draw[->,dashed,color=blue!30] (4.2,-3.7) -- (2.5,-3.7);
   
  \node at (1,-4.5) {$\text{(I)}$};
  
  \node at (4,-4.5) {$\text{(II)}$};

 \end{tikzpicture}
\caption{Structure of Models}
\label{fig:structure-of-models}
\end{figure}
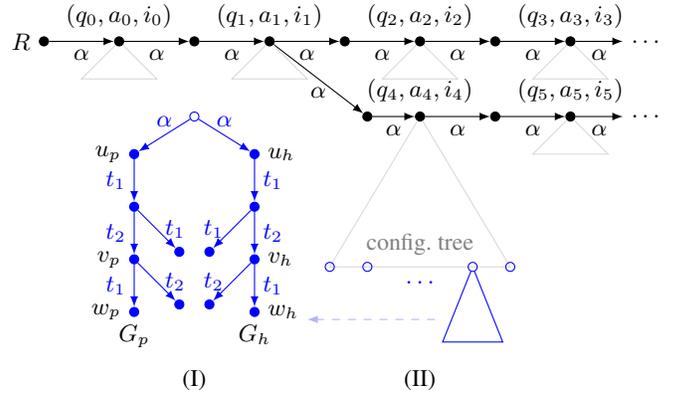

The intended models of the TBox $\Tmc_w$ are illustrated in Figure~\ref{fig:structure-of-models}.
The TBox $\Tmc_w$ is designed in such a way that its transitive-tree models rooted at an element satisfying $R$ encode a computation tree as follows. As indicated in the upper part of Figure~\ref{fig:structure-of-models}, every other node is labeled with a triple $(q,a,i)$ where $q\in Q$, $a\in\Gamma$ and $i\in [2^n] = \{0,1, \dots, 2^n-1\}$, suitably encoded using concept names. This tree structure is supposed to mimick the structure of the computation tree of \Mmc; every $(q,a,i)$-labeled node has further $\alpha$-successors which together encode the actual configuration of this node in the computation tree (as in Section~\ref{sec:lower}). The label $(q,a,i)$ provides the following abstract information about the stored configuration: 
\begin{itemize}
    \item $q$ is the state in the configuration,
    \item $a$ is the symbol under the  head, and 
    \item $i$ is the head's position.
\end{itemize}
%

We represent a configuration of \Mmc in the leaves of a \emph{configuration tree}, which is a binary tree of height~$n$, built from $\alpha$-edges. That is, each leaf is reached from the root by following a path of $n$ $\alpha$-edges. In Figure~\ref{fig:structure-of-models}, configuration trees are shown as gray triangles starting in  $(q,a,i)$-nodes. It is not difficult to axiomatize the described general structure of models in a TBox.\footnote{As our reduction draws from the one developed for \SH in~\cite{EiterLOS09} and the concrete construction of the TBox is rather standard given the underlying ideas, we refrain from giving the TBox here in its full details and rather describe its intended behavior.} 
The challenge, however, is to synchronize neighboring configurations. The main trick is that each configuration tree will represent \emph{two} configurations: the \emph{current} and the \emph{previous} one. The TBox ensures that in a given configuration tree, the represented current configuration is actually a successor configuration of the represented previous configuration. The query $q_w$ is then used to guarantee throughout the computation tree that the current configuration stored in a configuration tree equals the previous configuration stored in each of the successor configuration trees. We use the following concept names in the encoding: 
\begin{itemize}

  \item $B_0,\ldots,B_{n-1}$ to address the cells in a configuration,
  
  \item $X_q$ for each $q\in Q\cup\{\Box\}$, and $Y_a$ for each $a\in\Gamma$ to represent the contents of configurations, 
  
  \item auxiliary symbols $Z_{q,a}$ for $q\in Q\cup\{\Box\}, a\in \new{\Gamma}$, and
  
  \item additional concept names $G_p, G_h$ used as markers. 
  
\end{itemize}
\newcommand{\val}{\ensuremath{\mathsf{val}}\xspace}
(Note the differences with the naming convention in the main part of the paper.)
We write $\val_\Imc(d)$ for the number in $[2^n]$ encoded in binary by the interpretation of $B_0,\ldots,B_{n-1}$ at element $d\in \Delta^\Imc$. A tape cell with address $i$ and content $a\in \Gamma$ is represented by an element $d$ with $\val_\Imc(d)=i$ satisfying $Y_a$. If the head is currently placed over cell $i$ and the machine is in state $q$, then $d$ additionally satisfies $X_q$; otherwise, $d$ satisfies $X_\Box$.

We define now the central gadget, see Figure~\ref{fig:structure-of-models}~(I); the indices $\cdot_h$ and $\cdot_p$ used in the definition stand for \emph{here} and \emph{previous} (configuration).\footnote{This is an analogue of Definition~1 in~\cite{EiterLOS09}.} The gadget is similar to the one in Figure~\ref{fig:config-trees} used for UCQs in Section~\ref{sec:lower}, but it stores the current configuration and the previous one in separate branches to facilitate dealing with CQs. 

\begin{definition}[$i$-cell]\label{def:icell} Let \Imc be an interpretation and $i<2^n$. We call $d\in \Delta^\Imc$ an $i$-cell if the following hold: 
\begin{enumerate}

\item $d$ has $\alpha$-successors $u_p$, $u_h$ with $\val_\Imc(u_p)=\val_\Imc(u_h)=i$, each satisfying exactly one $X_q,q\in Q\cup\{\Box\}$, and exactly one $Y_a,a\in\Gamma$; 

\item $u_p$ and $u_h$ have $\alpha$-successors $v_p$ and $v_h$, respectively, such that $\val_\Imc(v_p)=\val_\Imc(v_h)$ is the bitwise complement of $i$ and for each $q\in Q\cup\{\Box\}$ and each $a\in \Gamma$, 
\begin{enumerate}[label=(\roman*)]

    \item \label{it:uh'_and_up} both $v_h$ and $u_p$ satisfy $Z_{q,a}$,
    
    \item \label{it:uh} 
    \new{$u_h$ satisfies $Z_{q,a}$ unless $u_h$ satisfies both $X_q$ and $Y_{a}$,}
    \item \label{it:up'} 
    \new{$v_p$ satisfies $Z_{q,a}$ unless $u_p$ satisfies both $X_q$ and $Y_{a}$;}
    
\end{enumerate}

\item $v_p$ and $v_h$ have $t_1$-successors $w_p$ and $w_h$ satisfying $G_p$ and $G_h$, respectively;


\item the elements between $u_p$ and $v_p$ and 
between $u_h$ and $v_h$ have (anonymous) $t_1$-successors, and $v_p$ and $v_h$ have (anonymous) $t_2$-successors, as shown in Figure~\ref{fig:structure-of-models}~(I).

\end{enumerate}   
\end{definition}

Hence, elements $u_p$ and $u_h$ store the actual content of the tape cells of the previous and the current configuration, respectively. Elements $v_p,w_p,v_h,w_h$ are needed for the synchronization later on. \new{Note that in items (i)-(iii) these elements are not treated symmetrically. For the current configuration, the cell content $(q,a)$ is encoded by the absence of $Z_{q,a}$ in $u_h$, whereas for the previous configuration, it is encoded by the absence of $Z_{q,a}$ in $v_p$, a child of $u_p$. This difference of one level, along with the negation involved in the encoding, is what will allow us to detect copying errors with a single CQ. }

We now formally define $(q,a,i)$-configuration nodes, the roots of configuration trees.\footnote{This is an analogue of Definition~2 in~\cite{EiterLOS09}. \new{Condition~3 is not part of the conference version, but is used in the technical report~\cite{EiterLOS09TR}}.} 
Below, by an \emph{$\alpha^m$-successor} we mean an element reachable by traversing $m$ $\alpha$-edges.

\begin{definition}[$(q,a,i)$-configuration node] 
\label{def:config-node}
An element $d$ of an interpretation $\Imc$ is a \emph{$(q,a,i)$-configuration node} if:
\begin{enumerate}

    \item \new{for every $j<2^n$, $d$ has an $\alpha^n$-successor $e_j$ that is a $j$-cell, referred to as \emph{the $j$-cell of $d$};}
    
    \item \new{the $u_h$-elements of the $j$-cells $e_j$, $j<2^n$, describe a configuration compatible with $q,a,i$, that is,}
    \begin{itemize}
    
        
        \item the $u_h$-element of $e_i$ satisfies $X_q$ and $Y_a$;
        
        \item the $u_h$-elements of all $e_j$ with $j\neq i$ satisfy $X_\Box$.
        
    \end{itemize}

    \item \new{if $q\neq q_0$, then the $u_p$-elements of the $j$-cells $e_j$, $j<2^n$, describe a configuration (in the way specified above) such that the configuration stored in the $u_h$-elements is a successor configuration of the one stored in the $u_p$-elements.}
    
\end{enumerate}

\end{definition}


\new{The reader might have noticed that Definition~\ref{def:config-node} does not mention any tree structure, unlike the intuition given at the beginning of the proof. Following \cite{EiterLOS09}, we include only those assumptions that are necessary to construct the query $q_w$. An auxiliary tree structure is used when axiomatizing configuration nodes with a TBox, but the query need not be aware of it. }

The intended interpretations encoding computations of $\Mmc$ are built from $(q,a,i)$-configuration nodes as follows.

\begin{definition}[computation tree]
\label{def:computation-tree}
A transitive-tree interpretation $\Imc$ is a \emph{computation tree for} $w$ if the following holds:
\begin{enumerate}

    \item the root $d_0$ of $\Imc$ has an $\alpha$-successor $d$ that is a $(q_0,a,0)$-configuration node whose $i$-cells describe the initial configuration for input $w$; 
    
    \item \label{it:succell}
    for each $(q,a,i)$-configuration node $d$, if $q\in Q_\exists$ (resp., $q\in Q_\forall$), then for some (resp., for all) transitions $(p,\b,\ell)\in \delta(q,a)$ there is an $\alpha^2$-successor $e$ of $d$ that is a $(p,\new{a'} ,i{+}\ell)$-configuration node \new{for some $a'\in\Gamma$}. 
\end{enumerate}
\end{definition}

We call the computation tree \Imc \emph{accepting} if $q=q_\mathsf{acc}$ in each $(q,a,i)$-configuration node without successor configuration nodes. Furthermore, $\Imc$ is \emph{proper} if for each pair of successive configuration nodes $d,e$ as in Item~\ref{it:succell} in Definition~\ref{def:computation-tree} and each $j<2^n$, the $j$-cell of $d$ has the same $(X_q,Y_a)$-label in its $u_h$-element as the $j$-cell of $e$ in its $u_p$-element\new{; that is, the configuration described in the $u_p$-elements of the $j$-cells of $e$ is the same as the one described by the $u_h$-elements in the $j$-cells of $d$.}

Based on the provided intuitions, it is routine to prove the following characterization of acceptance. 

\begin{proposition}
    $\Mmc$ accepts $w$ iff there is a proper accepting computation tree for $w$.
\end{proposition}

With some effort, and using fresh concept names to propagate information, one can construct a TBox that enforces computation trees \new{(not necessarily proper)} below elements satisfying $R$. \new{It is almost identical to what has been done in the original proof~\cite{EiterLOS09TR}, so we only sketch the differences. The first difference is that the original TBox uses a single non-transitive role to encode the structure, whereas we use the composition $\alpha$ of two transitive roles. It is straightforward to adjust the TBox to this change. The second difference is that our gadget encoding an $i$-cell has a slightly different structure. It is easy to adapt the original concept inclusions enforcing correct $i$-cells.}
\begin{proposition}
Given $w$, we can build in polynomial time a knowledge base $\Kmc_w=(\Tmc_w,\Amc_w)$ with $\Amc_w=\{R(a)\}$ whose tree-shaped models are exactly the accepting computation trees for $w$.
\end{proposition}

We use the query $q_w$ to detect computation trees that are not proper; that is, a computation tree $\Imc$ is proper iff $\Imc \not\models q_w$. We first provide a characterization of properness in terms of the concept names introduced above. Let $\mathbf{B}=\{B_0,\ldots,B_{n-1}\}$ and $\mathbf{Z} = \big \{Z_{q,a} \bigm | a\in \Gamma, q\in Q\cup\{\square\}\big\}$. We call cells $c,c'$ \emph{$A$-conspicuous} for a concept name $A$ if 
\begin{itemize}

    \item[$(\dagger)$] $A$ is true at $u_h$ in cell $c$ and at $u_p$ in $c'$, or 
    
    \item[$(\ddagger)$] $A$ is true at $v_h$ in cell $c$ and at $v_p$ in $c'$.
    
\end{itemize}
We have the following characterization.\footnote{This is the analogue of Proposition~4 in~\cite{EiterLOS09}; our proof is almost verbatim the same.}
\begin{lemma}\label{lem:conspicious}
A computation tree is not proper iff $(\ast)$ there exist cells $c,c'$ in successive configuration trees in $\Imc$ that are $A$-conspicuous for all $A\in\mathbf{B}\cup\mathbf Z$.
\end{lemma}
\begin{proof}
  First note that if $c,c'$ are cells of successive configurations in \Imc, then the conditions imposed on $\val_\Imc(\cdot)$ in Definition~\ref{def:icell} imply that $\val_\Imc(c)=\val_\Imc(c')$ iff for all $B\in\mathbf B$, $c$ and $c'$ are $B$-conspicuous; this is because bitwise complement is used for the addresses of the $v_p$ and $v_h$ elements in the cells. 

  Now suppose that $\Imc$ is proper and let $c,c'$ be cells of two successive configurations. First, if they are not $B$-conspicuous for some $B\in\mathbf B$, then $(\ast)$ is indeed violated. Now, assume that $c$ and $c'$ are $B$-conspicuous for every $B\in \mathbf B$. As explained above, this means that $\val_\Imc(c)=\val_\Imc(c')$. Now, as \Imc is proper, the $u_h$ element in $c$ and the $u_p$ element in $c'$ must satisfy the same $X_q$, for $q\in Q\cup\{\Box\}$, and the same $Y_a$, for $a\in \Gamma$. Then, by Item~2~\ref{it:uh} of Definition~\ref{def:icell}, $Z_{q,a}$ is false at the $u_h$ element of $c$; and by Item~2~\ref{it:up'}, $Z_{q,a}$ is false at the $v_p$ element of $c'$. 
  Hence, $c,c'$ cannot be $Z_{q,a}$-conspicuous, and~$(\ast)$ is violated as well.

  Conversely, let \Imc be improper. Then there exist two $j$-cells $c,c'$ in successive configurations such that the $u_h$ element of $c$ and the \new{$u_p$} element of $c'$ satisfy some different pairs $X_q,Y_a$ and $X_{q'},Y_{a'}$ of concept names. As $\val_\Imc(c)=j=\val_\Imc(c')$, $c$ and $c'$ are $B$-conspicuous for every $B\in \mathbf B$.
  Then, by Item~2~\ref{it:uh'_and_up} of Definition~\ref{def:icell}, $Z_{q,a}$ is true at the $v_h$ element of $c'$. Now, by Item~2~\ref{it:up'}, since $(q,a)\neq (q',a')$, $Z_{q,a}$ is also true at the $v_p$ element of $c'$. We can argue symmetrically that $Z_{q',a'}$ is true at the $u_h$ element of $c$ and the $u_p$ element of $c'$. For $(q'',a'')\notin\{(q,a),(q',a')\}$, $Z_{q'',a''}$ holds at the $u_p,v_p,u_h,v_h$ elements of both $c$ and $c'$. In summary, $c,c'$ are $Z$-conspicuous for all $Z\in\mathbf{Z}$ and thus $(\ast)$ is satisfied. 
\end{proof}

It thus remains to find the query $q_w$ that has a match into \Imc iff $(\ast)$ in Lemma~\ref{lem:conspicious} is satisfied. This is the crucial place where our reduction differs from the one in~\cite{EiterLOS09}. While the overall structure of our $q_w$ is the same as in~\cite{EiterLOS09}, the inner structure differs. 

The query $q_w$ is shown in Figure~\ref{fig:queries-lower}~(II). It is obtained by taking a copy of the query $q_A$ from Figure~\ref{fig:queries-lower}~(I) for each $A\in\mathbf B\cup\mathbf Z$ and identifying the variables satisfying $G_h$ and $G_p$, respectively.
In Figure~\ref{fig:queries-lower}~(I), the edge labeled with $\alpha^{n+2}$ represents a path with alternating $t_1$-edges and $t_2$-edges of total length $2(n{+}2)$, and likewise for $\alpha^{n+4}$. The dashed edge from $x_1^A$ to $y_1^A$ labeled with $q^*$ means that $x_1^A$ and $y_1^A$ are connected with the query $q^*$, shown in  Figure~\ref{fig:q-star},
by identifying $x$ with $x_1^A$ and $y$ with $y_1^A$; likewise for $x_2^A$ and $y_2^A$.

\begin{figure} 
\centering
\begin{tikzpicture} [>=latex,thin,point/.style={circle,draw=black,fill,minimum
    size=1.2mm,inner sep=0pt},scale=1.0, bend angle=45]\footnotesize

  \node[point, label=above:{$x^A$}] (a0) at (1,0) {};
  \node[point, label=left:$A$, label=right:$x^A_1$] (a1) at (0,-1) {};
  \node[point,label=below:$G_h$,label=right:$y^A_1$] (a2) at (0,-2) {};
  \node[point,label=right:$A$, label=left:$x^A_2$] (b1) at (2,-1.3) {};
  \node[point,label=below:$G_p$,label=left:$y^A_2$] (b2) at (2,-2.3) {};

  \draw[->] (a0) -> node[left] {$\alpha^{n+2}$} (a1);
  \draw[->,dashed] (a1) -> node[left] {$q^*$} (a2);
  \draw[->] (a0) -> node[right] {$\alpha^{n+4}$} (b1);
  \draw[->,dashed] (b1) -> node[left] {$q^*$} (b2);

  \node[point, label=above:{$x^A$}] (c0) at (4,0) {};
  \node[point, label=left:$A$] (c1) at (3,-1) {};
  \node[point, label=below:$G_h$] (c2) at (3,-2) {};
  \node[point, label=right:$A$] (d1) at (5,-1.3) {};
  \node[point, label=below:$G_p$] (d2) at (5,-2.3) {};

  \draw[->] (c0) -> (c1);
  \draw[->,dashed] (c1) -> (c2);
  \draw[->] (c0) ->  (d1);
  \draw[->,dashed] (d1) -> (d2);

  \node[point, label=above:{$x^B$}] (e0) at (4.7,0) {};
  \node[point, label=left:$B$] (e1) at (3.7,-1) {};
  \node[point, label=right:$B$] (f1) at (5.7,-1.3) {};

  \draw[->] (e0) -> (e1);
  \draw[->] (e0) ->  (f1);
  \draw[->,dashed] (e1) -> (c2);
  \draw[->,dashed] (f1) -> (d2);

  \node[point, label=above:{$x^Y$}] (g0) at (6,0) {};
  \node[point, label=left:$Y$] (g1) at (4.5,-1) {};
  \node[point, label=right:$Y$] (h1) at (7,-1.3) {};

  \node at (5.4,0) {$\ldots$};

  \draw[->] (g0) -> (g1);
  \draw[->] (g0) -> (h1);
  \draw[->,dashed] (g1) -> (c2);
  \draw[->,dashed] (h1) -> (d2);
  
  \node at (1,-3) {$\text{(I)}$};
  
  \node at (5,-3) {$\text{(II)}$};
  
  \end{tikzpicture}
    
\caption{Queries $q_A$ and final query $q_w$.}
\label{fig:queries-lower}
\end{figure}
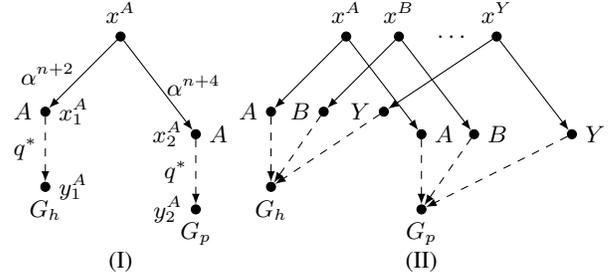

\begin{figure}
\centering
\begin{tikzpicture}[>=latex,thin,point/.style={circle,draw=black,fill,minimum size=1.4mm,inner sep=0pt},scale=1.0, bend angle=45]\footnotesize

\node[point, label=left:{$x$}] (x0) at (6,0) {};
  \node[point] (y0) at (7,-1) {};
  \node[point] (x1) at (6,-0.8) {};
  \node[point] (y1) at (7,-1.8) {};
  \node[point] (x2) at (6,-1.6) {}; \node[point,label=left:{$y$}] (y2) at (6,-2.3) {};

  \draw[->] (x0)-> node[right]{$\:t_1$} (y0);
  \draw[->] (x1)-> node[above]{$t_1\;\;$} (y0);
 \draw[->] (x1)-> node[right]{$\:t_2$} (y1);
 \draw[->] (x2)-> node[above]{$t_2\;\;$} (y1);
 \draw[->] (x2)-> node[right]{$t_1$} (y2);
    
\end{tikzpicture}

%
%
\caption{Query $q^*$.}
\label{fig:q-star}
\end{figure}
%

%
\begin{proposition}
   A computation tree \Imc is proper iff $\Imc\not\models q_w$.
\end{proposition}
\begin{proof}
We claim that a match of $q_A$, $A\in \mathbf B\cup \mathbf{Z}$, into some computation tree $\Imc$ detects $A$-conspicuous cells in two successive configurations. Observe first that $y_1^A$ and $y_2^A$ have to be mapped to the ends of two cells $c,c'$, due to concept names $G_h$ and $G_p$. Moreover, these cells have to be in successive configuration trees since the path from $x^A$ to $x_1^A$ is precisely $\alpha^2$ shorter than the path from $x^A$ to $x_2^A$. Finally, due to the query $q^*$ between $x_1^A$ and $y_1^A$ (and $x_2^A$ and $y_2^A$, respectively) there is a choice as to where $x_1^A$ and $x_2^A$ are mapped:
\begin{itemize}
    \item either $x^A$ is mapped to the root of the configuration tree of cell $c$, \new{in which case} $x_1^A$ is mapped to the $v_h$ element of its cell, and $x_2^A$ is mapped to the $v_p$ element in its cell, or

    \item $x^A$ is mapped to the $\alpha$-predecessor of the root of the configuration tree of cell $c$, 
    $x_1^A$ is mapped to the $u_h$ element of its cell and $x_2^A$ is mapped to the $u_p$ element in its cell.
    
\end{itemize}
This choice is the crucial insight here. It is enabled by the structure of $q^*$ and the anonymous elements in the cells provided by Definition~\ref{def:icell}~(4). The first case corresponds to case~$(\ddagger)$ in the definition of $A$-conspicuous cells while the second case corresponds to case~$(\dagger)$. 

The proposition follows immediately from the claim. 
\end{proof}

This finishes the proof of Theorem~\ref{thm:2expComp}.

\section{Proof of Theorem~\ref{th:acyclic_transitive}}
\label{app:acyclic}

This section of the appendix is devoted to the details of the proof of Theorem~\ref{th:acyclic_transitive} in Section~\ref{sec:acyclic}. In particular, we prove that the constructed interpretation $\Jmc_{d_0}$ satisfies the desired conditions of the theorem.

To this end, we show that each interpretation $\Jmc_d$ satisfies the conditions of the invariant:
\begin{itemize}

    \item $d \in \Delta^{\Jmc_d}\subseteq \Delta^\interp$, 
    
    \item $\tp(\Jmc_{d}, e)=\tp(\interp, e)$ for all $e\in\Delta^\Jmc_{d}$,
    
    \item $\rch^{\Jmc_{d}^+\!}_t(e) = \rch^{\Imc}_t(e)$ for all $e\in\Delta^\Jmc_{d}$,
    
    \item $\Jmc_d^+ \not \models \{q_{d,i} \mid i\leq m\}$.
    
    \item  $|\Delta^{\Jmc_d}|\leq (|\rch^{\Imc}_t(d)|+1)!$, 
    
\end{itemize}

The first three conditions follow directly from the construction. The remaining two are proved in the following two lemmas. Notice that, for the case where $d=d_0$, these two conditions become respectively $\Jmc_{d_0}^+\not\vDash Q$ and $|\Delta^{\Jmc_{d_0}}|\leq (|\NC(\tbox)|+1)!$ (since $\rch^{\Jmc_{d_0}^+}_t\subseteq\NC(\tbox)$), as required by the theorem.

\begin{lemma}
    $|\Delta^{\Jmc_d}|\leq (|\rch^{\Imc}_t(d)|+1)!$.
\end{lemma}

\begin{proof}
In both cases, the interpretation $\Jmc_{d}$ is obtained by putting together at most $|\rch^{\Imc}_t(d)|+1$ pieces: each of these pieces being either a single element, or interpretation of the shape $\Jmc_e$, with $|\rch^{\Imc}_t(e)| < |\rch^{\Imc}_t(d)|$. Hence, this upper bound follows by a routine induction.
%
%
\end{proof}






\begin{lemma}
$\Jmc_d^+ \not \models \{q_{d,i}\mid i\leq m\}$.
\end{lemma}

\begin{proof} 
    We proceed by induction on $\rch^\Imc_t(d)$. In the inductive base, $\rch^\Imc_t(d)=\emptyset$ and $\Jmc_d$ consists just of the single element $d$ which is a copy of $d$ in $\Imc_d$. Since, $\Imc_d\not\models  \{q_{d,i}\mid i\leq m\}$, the same is true for $\Jmc_d$.
    
    For the inductive step, let $\rch^\Imc_t(d)\neq \emptyset$ and suppose that $\Jmc_e^+ \not \models \{q_{e,i}\mid i\leq m\}$ for all $e \in \Delta^\Imc$ such that $\rch^\Imc_t(e)\subsetneq \rch^\Imc_t(d)$.
    Towards showing a contradiction, assume that $q_{d,i}$, for some $i\leq m$, admits a match $\delta$ in $\Jmc_d^+$.

   Regardless of the case in the construction we followed, each element in $\Jmc_d^+$ is reachable from $d$. In consequence, using the notation from the definition of $\bar \qbf$, we can assume without loss of generality that $\delta(x)=d$ for all $x\in X_{d,i}$. Consider the (nonempty) subquery $q'_{d,i}$ of $q_{d,i}$, as in the definition of $\bar \qbf$.
   By the anti-monotonicity of $\preceq$, $q_{e,i}$ is a subquery of
   $q'_{d,i}$, for all $e\in \desc{d}$. 
   

    Let us consider the first case of the dichotomy, illustrated in Figure~\ref{fig:dicho} on the left. Because $q_{d,i}'$ is connected and none of its initial variables are mapped to $d$, it follows that $\delta$ is a match of $q_{d,i}'$ in  $\Jmc_{f_j}^+$ for some $j\leq \ell$. Since $f_j\in \desc{d}$, 
    $q_{f_j,i}$ is a subquery of $q_{d,i}'$ (as noted above). Hence, $\delta$
   is also a match of $q_{f_j,i}$ in $\Jmc_{f_j}^+$, that is, $\Jmc_{f_j}^+\models q_{f_j,i}$. This is, however, in contradiction to the induction hypothesis. 

    Next, consider the second case of the dichotomy, illustrated in Figure~\ref{fig:dicho} on the right. 
    By construction, $q_{e_1,i} = \dots = q_{e_k,i}$. 
    Let us call this query $q_{d,i}''$. 
    Because all $e_j$ belong to $\desc{d}$, it follows from the anti-monotonicity of $\preceq$ that $q_{d,i}''$ is a subquery of $q_{d,i}'$, so $\delta$ matches $q_{d,i}''$ in $\Jmc_d^+$. By the choice of the $e_j$ (in particular the fact that the $\bar q$ is stable in the subtrees rooted at $e_j$), we know that $X_{e_j,i}=\emptyset$, for all $j\leq k$, that is, the initial variables of $q_{e_j,i}(=q_{d,i}'')$ do not match in $e_j$. 
    %
    Because $q_{d,i}''$ is connected, this implies that $\delta$ is a match of $q_{d,i}''$ in $\Jmc_{d_j}^+$ for some $j\leq \ell$. 
    But again, this contradicts the induction hypothesis, because $q_{d_j}$ is a subquery of $q''_{d,i} = q_{e_j,i}$, due to the anti-monotonicity of~$\preceq$.

    We can conclude: no matter which case of the construction was applied, no $q_{d,i}$ admits a match in $\Jmc_d^t$.
\end{proof}

\section{Appendix of Section~\ref{sec:ptqs}}
\label{ap:ptqs}

\new{
This  appendix is devoted to the proofs of the results of Section~\ref{sec:ptqs}, and in particular, it provides the tools needed in order to demonstrate Lemma~\ref{lem:to_ptq}.}
%
We need one auxiliary notion, that is used in this proof, but also in two other places later on; so we introduce it first, and prove a crucial property. \new{Recall that Lemma~\ref{lem:to_ptq} applies to CQs that use a single transitive role $t$.} 
The following definition is similar to what was called \emph{fork elimination} in~\cite{EiterLOS09}.
\begin{definition}\label{def:queryequivalence}
Given a unary or Boolean conjunctive query $q(\bar x)$, we define $\approx$ to be the smallest equivalence relation on $\var(q)$ satisfying the following conditions:
\begin{itemize}[leftmargin=*]
 
   \item[$(\dagger)$] If $q$ contains atoms $r(x,z),r(y,z')$ for $r\neq t$
   and $z\approx z'$, then $x\approx y$. 
   
  \item[$(\ddagger)$] If $q$ contains atoms $r(x,y)$ and $s(x',y')$ with $t\notin\{r,s\}$ and $y,y'\in \var(C)$ for some $t$-cluster $C$, then $y\approx y'$. 
  
\end{itemize}
We define $q_\approx$ as the query obtained from $q$ by identifying equivalent variables. We further denote with $[x]$ the equivalence class of variable $x$; thus, the classes $[x]$, $x\in\var(q)$ are the variables of $q_\approx$.
\end{definition}

   Recall that $\mathbf T$ is the class of all transitive-tree interpretations. \new{Let us prove some basic properties of $q_\approx$.}
\begin{lemma}\label{lem:queryequivalence}
\mbox{} 
\begin{enumerate}[leftmargin=*]

    \item For any Boolean CQ $q$, we have $\Imc\models q$ iff $\Imc\models q_\approx$ for all $\Imc\in\mathbf{T}$. 
    
    \item For any unary CQ $q(x)$, we have $\pair{\Imc}{d}\models q(x)$ iff $\pair{\Imc}{d}\models q_\approx([x])$ for all $\Imc\in\mathbf{T}$. 

    \item For any unary CQ $q(x)$ such that 
    $\pair{\Imc}{d}\models q(x)$ for some $\Imc\in \mathbf T$ with root $d$, no variable in $q_\approx$ is reachable from variables $[y]$ and $[y']$ such that $r([x],[y]), s([x],[y'])\in q_\approx$ for some $r\neq s$.
    
   
\end{enumerate}
     
\end{lemma}

\begin{proof}
Items~1 and~2 follow from the fact that in every match $\delta$ of $q$ in $\Imc$ for some $\Imc\in\mathbf T$, we have $\delta(x)=\delta(y)$, for all $x\approx y$. This can easily be proved by induction on the definition of $\approx$. Item~3 is a consequence of the definition of $\approx$. Indeed, if there were two such paths, then a straightforward inductive argument shows that there is a variable that has two incoming role atoms \new{involving different role names}, which is impossible if $\pair{\Imc}{d}\models q(x)$ for some $\Imc\in\mathbf T$.
\end{proof}

   \toPTQ*

\new{The assumption that the query uses at most one transitive role name is not necessary for the decidability in polynomial time of the existence of a match in a transitive-tree interpretation, but the argument we use below to show it  does exploit the assumption in order to simplify the construction of suitable PTQ(s). We provide a general argument in Lemma~\ref{lem:matchtotreeshaped} at the end of this section.} 

\begin{proof}
We give the proof for Boolean CQs first and comment on the necessary changes in the end. Consider a connected Boolean CQ $q$ with at most one transitive role, say~$t$, and let $q'=q_\approx$ as defined in Definition~\ref{def:queryequivalence}. By Lemma~\ref{lem:queryequivalence}~(1), 
$\Imc\models q$ iff $\Imc\models q'$ for all $\Imc\in\mathbf{T}$. So, we proceed with $q'$ (which can clearly be computed in polynomial time). 
%
%
%
%

Let $C_1,\ldots,C_k$ be the clusters in $q'$, corresponding to role names $r_1,\ldots,r_k$, respectively. We define a binary relation $\leadsto$ on these clusters by defining, for all $i\neq j$, $C_i\leadsto C_j$ if there is some $x\in \text{var}(C_i)\cap\text{var}(C_j)$ such that:
\begin{itemize}

\item $x$ is initial in $C_j$, and 

\item \new{if $r_i$ is non-transitive then $x$ is not initial in $C_i$.}

\end{itemize}
The definition of $\leadsto$ is 
\new{designed} 
to mimic the intuition about cluster trees provided in the main part. Intuitively, $C_i\leadsto C_j$ 
\new{means that $C_j$ is mapped `below' $C_i$ by matches to transitive-tree interpretations.} 
In this spirit, $\leadsto$ will be the starting point when constructing 
\new{a}
cluster tree for $q'$ later on in this proof. 

Let us first characterize when $q'$ has a match into a \new{transitive-tree} interpretation.

\medskip\noindent
\textbf{Claim~1.} $q'$ has a match into a transitive-tree interpretation iff the following conditions are satisfied: 
 \begin{enumerate}[label=(\roman*),leftmargin=*]
 
     \item \label{it:acyclicity} $q'$ is acyclic;
 
     \item \label{it:incoming} $q'$ does not contain atoms of shape $r(x,z),s(y,z)$ with $r\neq s$;
     
     \item \label{it:leadsto_acyclicity} the relation $\leadsto$ is acyclic.
     
 \end{enumerate}
 
 \medskip\noindent \textit{Proof of Claim~1.} We start with the ``only if'' direction. If any of Conditions~\ref{it:acyclicity} and~\ref{it:incoming} is not satisfied, then clearly $q'$ does not have a match into a \new{transitive-tree interpretation}. Suppose now that 
 Condition~\ref{it:leadsto_acyclicity} is not satisfied, in order to show the same conclusion.
 
 Let us consider a sequence of clusters $D_1,\ldots,D_{m-1}$ with $m>2$ that forms a cycle with respect to $\leadsto$, that is, $D_{1}\leadsto \ldots \leadsto D_{m-1}\leadsto D_{1}$. For the sake of convenience, we define $D_m=D_1$. Let $s_1,\ldots,s_{m-1}$ be the roles of the clusters, and let $x_1,\ldots, x_{m-1}$ be the variables witnessing $\leadsto$, that is: for each $j<m$, $x_j$ belongs to $\text{var}(D_j)\cap\text{var}(D_{j+1})$, is initial in $D_{j+1}$, and not initial in $D_j$ if $s_j$ is non-transitive.
 
 Notice that one of $s_1,\ldots,s_{m-1}$ has to be non-transitive. Otherwise all $D_j$ would be $t$-clusters, 
 \new{which is impossible because different $t$-clusters never share variables.}\footnote{When two transitive roles are allowed, there might be no non-transitive cluster involved. And indeed, this is the point where an attempt to prove the claim for more than one transitive roles would fail. \new{For example, the naughty query from Figure~\ref{fig:naughty}~(left) can be matched in a transitive-tree interpretation, but the corresponding $\leadsto$ is not acyclic.}} Without loss of generality, we assume that $s_1$ is non-transitive. 
 
  As $s_1$ is not transitive, $x_1$ is not initial in $D_1$, and there is an atom $s_1(x_0,x_1)\in D_{1}$, with $x_0$ being the unique initial variable in $D_1$. But since $x_{m-1}$ is also initial in $D_m=D_1$, we have $x_{m-1}=x_0$.
 
 Assume now that $q'$ admits a match $\delta$ into a \new{transitive-tree interpretation} $\interp$. Because of the atom $s_1(x_0,x_1)$, we have $\pair{\delta(x_0)}{\delta(x_1)}\in s_1^{\interp^+}$, and by simple induction we can show that $\delta$ maps every variable $x_j$ for $1 \leq j \leq  m-1$ strictly below $\delta(x_0)$, thus reaching contradiction with the assertion $x_{m-1}=x_0$. This concludes the proof of the ``only if'' direction.

\medskip 

Let us now prove the ``if'' direction. We suppose that $q'$ satisfies Conditions~\ref{it:acyclicity}--\ref{it:leadsto_acyclicity}, and we show that $q'$ has a match to a \new{transitive-tree interpretation}.

We define from $q'$ a new CQ $q''$ where each $t$-cluster is ``linearized'' in a way described as follows. Consider any $t$-cluster $C$ in $q'$. Since, by Condition~\ref{it:acyclicity}, $q'$ is acyclic so is the cluster $C$. Let $\ell\geq 1$ be the length of the longest path in $C$. Due to acyclicity, we can assign \emph{levels} $0,\ldots,\ell$ to the variables in $C$ in such a way that $t(x,y)\in C$ implies that the level of $x$ is lower than the level of $y$, and that the initial variables in $C$ have level~$0$.\footnote{There might be several assignments satisfying these conditions, but it is sufficient for us to consider any of them.} 

First, $q''$ contains every atom of $q'$ of the shape $r(x,y)$ (resp. $A(x)$), where $x$ and $y$ are \emph{not} variables of any $t$-cluster of $q'$. Second, for every $t$-cluster $C$, with longest path of length $\ell\geq 1$:

    \begin{enumerate}
        \item $q''$ contains the chain of atoms $t(x^C_0,x^C_1),\ldots,$ $t(x^C_{\ell-1},x^C_\ell)$, with fresh variables $x^C_0,\ldots, x^C_{\ell}$.

        \item If there exists in $q'$ an atom of the shape $r(y,x)$ with $x\in\var(C)$ and~$r\neq t$, then $q''$ contains the atom $r(y,x^C_0)$. (Note that, due to 
        \new{Condition}~$(\ddagger)$, there is at most one such $x$, and by Condition~(ii) it is initial in $C$.) 

        \item For each atom $r(x,y)$ (resp. $A(x)$), with $r\neq t$ and $x\in\var(C)$ of level $i\leq \ell$, $q''$ contains the atom $r(x^C_i,y)$ (resp. $A(x^C_i)$).
    \end{enumerate}
Let $\Imc$ be 
the query $q''$ viewed as an interpretation in the standard way. By construction, $q'$ has a match to the transitive closure of $\Imc$. It remains to argue that $q''$ (and hence \Imc) is tree-shaped. Observe that each variable $x$ in $q''$ has at most one incoming atom $r(y,x)$, due to Condition~(ii) and 
\new{Conditions}~$(\dagger)$ and~$(\ddagger)$. So it suffices to show that $q''$ does not have a cycle. 
Suppose otherwise and consider a cycle, that is, a sequence of atoms $r_1(x_1,x_2),r_2(x_2,x_3),\ldots,r_n(x_n,x_1)$ in $q''$. Let $C_1,C_2,\ldots,C_n$ be the clusters containing these atoms and set, for convenience, $C_{n+1}=C_1$. By definition of $\leadsto$, it is the case that $C_i\leadsto C_{i+1}$ for all $i\leq n$, hence $\leadsto$ is not acyclic in contradiction to Condition~(iii).

%
%

This finishes the proof of Claim~1.\hfill $\dashv$

\bigskip
Clearly, Conditions~\ref{it:acyclicity}--\ref{it:leadsto_acyclicity} can be checked in polynomial time. This finishes the proof of the first part of the lemma \new{of Boolean CQs}.

\bigskip
It remains to argue that $q'$ is a PTQ if Conditions~\ref{it:acyclicity}--\ref{it:leadsto_acyclicity} from the claim
are satisfied. Condition~\ref{it:acyclicity} implies that Condition~\ref{it:cluster_acyclicity} of Definition~\ref{def:ptq} (of PTQs) is satisfied. To show that also Condition~\ref{it:cluster_tree} of Definition~\ref{def:ptq} is satisfied, we carefully arrange the clusters
$C_1,\ldots,C_k$ of $q'$ 
in a cluster tree satisfying \new{Items}~(a)--(c) of Definition~\ref{def:clustertree} (of cluster trees). As a starting point for this, we define $C_j$ as \emph{child} of $C_i$ if $C_i\leadsto C_j$ and one of the following  holds:
\begin{enumerate}[label=(\Alph*),leftmargin=*]

    \item $C_j$ is a $t$-cluster;


    \item both $C_i$ and $C_j$ are not $t$-clusters;

    \item $C_i$ is a $t$-cluster, $C_j$ is not a $t$-cluster and there is no $r$-cluster $D$ with $r\neq t$ and $D\leadsto C_j$.
    
\end{enumerate}

\noindent \new{Note that, in (C), such a cluster $D$ might exist if the variable shared by $C_i$ and $C_j$ is initial in $C_i$.}

\medskip\noindent\textbf{Claim~2.} Every cluster is the child of at most one other cluster. 

\medskip\noindent\textit{Proof of Claim~2.}  Suppose first that $C$ is some $t$-cluster. If it is the child of some other cluster $D$, then \new{it is} due to (A). 
Assume that there are \new{two different clusters,} $D_1$ and $D_2$, with $D_i\leadsto C$ for $i=1,2$, and let $x_i$ be the shared variable witnessing the latter. Note that both $D_i$ have to be $r_i$-clusters for some $r_i\neq t$, \new{because $t$-clusters never share variables}. 
\new{By the definition of $\leadsto$, $x_i$ cannot be initial in $D_i$. Hence, there are atoms $r_1(y_1,x_1),r_2(y_2,x_2)$ in $D_1,D_2$. From \new{Condition}~$(\ddagger)$ it follows that $x_1=x_2$. If $r_1 \neq r_2$, we get a contradiction with Condition~\ref{it:incoming}. If $r_1 = r_2$, by Condition~$(\dagger)$, $y_1=y_2$, which is impossible because two different $r_1$ clusters cannot share the initial variable. }

Suppose now that $C$ is an $r$-cluster for $r\neq t$. Then only (B) or (C) can make $C$ a child of other clusters. In case (C) applies, (B) does not apply and $C$ is the child of a unique cluster. If (B) applies, then due to Condition~\ref{it:incoming} and the definition of $\leadsto$, \new{it is impossible to find two different clusters $D_1$ and $D_2$ such that $D_i$ is an $r_i$-cluster with $r_i\neq t$ and $D_i\leadsto C$ for $i=1, 2$.}

This finishes the proof of Claim~2.\hfill $\dashv$

\medskip Hence, the child relation forms a forest, and it is not difficult to verify that Conditions~(a)--(c) of Definition~\ref{def:clustertree} of cluster trees are satisfied. We now turn this forest over clusters into a tree. Suppose the forest consists of trees $T_1,\ldots,T_m$ of clusters. We denote with $V_i$ the set of all variables that occur in (a cluster in) $T_i$. Clearly, the $V_i$ cannot be pairwise disjoint since $q'$ is connected. So suppose distinct $V_i,V_j$ share some variable $x$. Note that $x$ cannot appear as a non-initial variable in two different clusters, due to Condition~(ii). So let us suppose without loss of generality that $x$ occurs only as an initial variable in $T_i$, and let $D$ be any cluster in $T_i$ that mentions $x$.

We argue that $D$ is not a $t$-cluster. Towards a contradiction, assume that $D$ is a $t$-cluster and let $C$ be a cluster in $T_j$ mentioning $x$. Then, $C$ is an $r$-cluster for some $r\neq t$, \new{because different $t$-clusters never share variables. We now show that $C$ and $D$ are siblings, or one is a child of the other, either of which is impossible because $T_i$ and $T_j$ are disjoint.   
If $x$ is non-initial in $C$, then $C\leadsto D$ and hence $D$ is a child of $C$ by (A). Suppose that $x$ is initial in $C$. Then, $D\leadsto C$. If there is no $s$-cluster $D'$ with $s\neq t$ and $D'\leadsto C$, then $C$ is a child of $D$ by (C). If such $D'$ does exists, then both $D$ and $C$ are children of $D'$. We conclude that $D$ cannot be a $t$-cluster.}

Hence, $D$ is an $r$-cluster for $r\neq t$. Using 
\new{a very similar argument,} 
we can show that every cluster $C$ in $T_j$ mentioning $x$ is not a $t$-cluster. So let $C$ be any such cluster in $T_j$. Clearly, $x$ has to be initial in both $D$ and $C$, as otherwise either $D$ would be child of $C$ or vice versa. But then both $D$ and $C$ are also root clusters of $T_i$ and $T_j$, respectively. \new{Indeed, if either had a parent, the other one would be its sibling. }


Overall, we have shown that, if there is a variable shared between different $T_i,T_j$, then it has to be the initial variable of the root cluster in both $T_i$ and $T_j$, \new{neither of which is a $t$-cluster.} Clearly, there cannot be two such variables, due to connectedness of $q'$. Hence, $T_1,\ldots,T_m$ share a single variable $x$ with the mentioned properties. 
We then make the root cluster of $T_2,\ldots,T_m$ children of the root cluster of $T_1$. (The choice of $T_1$ is arbitrary, any other $T_i$ would work as well.) It can be shown that Conditions~(a)--(c) of Definition~\ref{def:clustertree} are preserved by this operation. This finishes the proof of the Boolean case.

\bigskip 
It remains to discuss the unary case. Let $q(x)$ be a unary CQ with $x$ initial. Let $q'(x')=q_\approx([x])$, where $q_\approx$ as in Definition~\new{\ref{def:queryequivalence}}. By Item~2 of Lemma~\ref{lem:queryequivalence}, we have $\pair{\Imc}{d}\models q(x)$ iff $\pair{\Imc}{d}\models q'(x')$ for all $\Imc\in\mathbf T$. One can show along the lines of Claim~1 that $\pair{\Imc}{d}\models q'(x')$ for some $\Imc\in\mathbf T$ with root $a$ iff $q'$ satisfies Conditions~(i)--(iii) from Claim~1 and $x'$ is initial in $q'$. 

Let $q_1(x'),\ldots,q_k(x')$ be all queries that can be obtained from $q'(x')$ by the following process: 

\begin{itemize}

    \item take all unary atoms $A(x')\in q'$;

    \item pick an atom $r(x',y)$ and compute the set $V$ of variables reachable from $y$ in the directed graph underlying $q'$;

    \item include in $q_i$ the atom $r(x',y)$ and all atoms that mention only variables from $V\cup\{x'\}$.
    
\end{itemize}
Item~3 of Lemma~\ref{lem:queryequivalence} implies that for all $\Imc\in\mathbf T$, $\pair{\Imc}{d}\models q'(x')$ iff $\pair{\Imc}{d}\models q_i(x')$ for all $i\leq k$. 

It remains to define the child relation as before (on each individual $q_i$) and observe that Claim~2 is not affected by having an answer variable $x'$. By construction, $x'$ is initial in each of the $q_i(x')$ and it is contained in a uniquely determined root cluster. Hence, each $q_i(x')$ is a PTQ, and we can return $q_1(x'),\ldots,q_k(x')$.
%

\end{proof}

\begin{lemma}\label{lem:matchtotreeshaped}
    There is a polynomial time algorithm that decides whether a given connected (unary or Boolean) CQ has a match to a \new{transitive-tree interpretation}.
\end{lemma}

\begin{proof}
Consider a connected (unary or Boolean CQ) $q$. 
We first transform $q$ by applying the following steps exhaustively:
\begin{enumerate}[label=(R\arabic*),leftmargin=*]
 
   \item If $q$ contains distinct atoms $r(x,z),r(y,z)$ for non-transitive $r$, then identify $x$ and $y$.
     

  \item If $q$ contains a $t$-cluster, $t$ transitive, with distinct initial variables $x,y$, then identify $x,y$.
  
 \end{enumerate}
Let $q'$ be the result of the transformation.
Note that the transformation can be done in polynomial time. 

\begin{claim}
\label{cl:charac_q'}
    $q$ has a match to a \new{transitive-tree interpretation} iff $q'$ satisfies the following conditions: 
 \begin{enumerate}[label=(\roman*),leftmargin=*]
 
     \item $q'$ is acyclic;
 
     \item $q'$ does not contain atoms $r(x,z),s(y,z)$ with $r\neq s$.

 \end{enumerate}
\end{claim}

\textit{Proof of the claim.} The "if"-direction is rather straightforward. The important insights are that (a) there is a match of $q$ to $q'$ (viewed as interpretation) and (b) we can transform $q'$ to a TQ $q''$ by "linearizing" the transitive clusters and such that $q'$ has a match to $q''$ (viewed as interpretation). See the proof of Lemma~\ref{lem:to_ptq} for details on the linearization.

For "only if", we show that for every (unary or Boolean) CQ $p$ we have that $p$ has a match to a \new{transitive-tree interpretation} if the query $p'$ obtained from $p$ by applying~(R1) or~(R2) has. This is clear for~(R1), since any match of $p$ containing $r(x,z),r(y,z)$ with $r$ non-transitive into a \new{transitive-tree interpretation} has to map $x$ and $y$ to the same element. 

For~(R2), suppose first that $p'$ has a match to a transitive-tree interpretation~\Imc. Since $p'$ is obtained from $p$ via variable identification, there is also a match of $p$ to $\Imc$. Suppose now that $p$ has a match $\rho$ to $\Imc^+$ for some tree-shaped interpretation~\Imc, and let $x,y$ be the variables in the $t$-cluster $C$ identified by~(R2). Let moreover denote $V$ the set of variables $V\subseteq \var(C)$ that are reachable in $C$ from $y$ but not from $x$ (including $y$). Consider the restriction $\Jmc$ of $\Imc$ to domain $\{\delta(z)\mit z\in V\}$. Clearly, $\Jmc$ is a tree. We transform $\Jmc$ into a path $\widehat\Jmc$ by identifying domain elements having the same distance from the root. Then we modify $\Imc$ as follows: 
\begin{itemize}

    \item merge the root of $\widehat\Jmc$ with $\delta(x)$ in $\Imc$; 
    
    \item make every successor of $\delta(x)$ a successor of the last element of $\widehat \Jmc$ (these successors are then not direct successors of $\delta(x)$ anymore);

    \item for every non-$t$-successor $e$ of some element $d$ in $\Jmc$ make $e$ a successor of the element in $\widehat J$ that identifies with $d$. 
    
\end{itemize}
Let $\widehat\Imc$ be the result of this transformation.
It is routine to verify that $p'$ has a match to $\widehat\Imc^+$.

\smallskip It remains to note that Conditions~(i) and~(ii) are obviously necessary conditions for having a match to a \new{transitive-tree interpretation}. 
\hfill$\dashv$

\medskip Since $q'$ can be computed in polynomial time and the conditions in Claim~1 can be decided in polynomial time as well, this finishes the proof of the lemma.
\end{proof}

\section{Appendix of Section~\ref{sec:upper}}
\label{ap:upper}

\new{This Appendix contains the missing proofs of the various results in Section~\ref{sec:upper}.}

\subsection{Proofs of Theorems~\ref{th:reduction-rooted} and~\ref{th:reduction-single}} \label{sap:reduction-single}

In order to prove Theorems~\ref{th:reduction-rooted} and~\ref{th:reduction-single}, we first provide a characterization of non-entailment of UCQs in terms of non-entailment of queries in the form of rooted entailment, see Propositions~\ref{pr:reduction-rooted} and~\ref{pr:reduction-single} below. Both the nature of the characterization and the tools used to prove it are relatively standard. Similar techniques can be found, for example, in~\cite{Lutz08}. 

First, we introduce the standard notion of forest interpretation: if $\Theta$ is a finite subset of $\NI$, then an interpretation $\interp$ is a \emph{$\Theta$-forest interpretation} if there is a family $(\Imc_a)_{a\in\Theta}$ of tree-shaped interpretations rooted in $a$ such that (i) the $(\Imc_a)_{a\in\Theta}$ have pairwise disjoint domains, and (ii) $\Imc$ is obtained as the transitive closure of the union of the restriction of $\Imc$ to $\Theta$ and all $\Imc_a$. Note that in all forest interpretations, transitive role names are interpreted as transitive relations. Furthermore, the tree-shaped interpretations $\Imc_a$ are uniquely determined by $\Imc$ and $\Theta$. We refer to them as \emph{induced by $\Imc$}.

It is known that UCQ entailment can be studied over forest interpretations without loss of generality, see for example~\cite{GlimmLHS08}.

\begin{restatable}{lemma}{forest}
\label{lem:forest}
  Let $\Kmc = \pair{\tbox}{\abox}$ be a knowledge base. Then: 
  \begin{enumerate}
      \item For each Boolean UCQ $Q$ with $\Kmc\not\models Q$, there is an $\indiv(\abox)$-forest interpretation $\Jmc$ 
  such that $\Jmc^+\vDash\Kmc$ but $\Jmc^+\not\models Q$.

  \item For each tuple $\bar a$ from $\NI(\Amc)$ and each rooted UCQ $Q(\bar x)$ with $\Kmc\not\models Q(\bar a)$, there is an $\indiv(\abox)$-forest interpretation $\Jmc$ 
  such that $\Jmc^+\vDash\Kmc$ but $\pair{\Jmc^+}{\bar a}\not\models Q(\vec x)$.
  \end{enumerate}
\end{restatable}

We now develop query decompositions, with the intention of decomposing possible matches of a CQ to forest interpretations. We need two auxiliary notions. A \emph{subdivision} of a CQ $q$ is any CQ that can be obtained by replacing an arbitrary number of atoms $t(x,z)$ with $t$ transitive by two atoms $t(x,y),t(y,z)$, for some fresh variable $y$. For a set $\Theta$,
a \emph{$\Theta$-split of a CQ $q$} is a family of pairs $(U_a,V_a)_{a\in \Theta}$ such that the 
following properties are satisfied, for each $a\in\Theta$: 
\begin{enumerate}[label=(S\arabic*),leftmargin=*]

  \item $\bigcup_{a\in\Theta}U_a\cup V_a =  \text{var}(q)$; 

  \item the sets $U_a,V_a$ are disjoint and disjoint from all other $U_b,V_b$ for $a\neq b\in\Theta$;

  \item for each $x\in V_a$ and $r(x,y)\in q$, we have $y\in V_a$;

  \item for each $x\in U_a$ and $r(x,y)\in q$, we have $y\in V_a$ or $y\in U_b$ for some $b\in \Theta$ with $b\neq a$.

\end{enumerate}

For a subset $V\subseteq \text{var}(q)$, let us denote by $q|_V$ the restriction of $q$ to variables in $V$. Given a $\Theta$-split $\sigma=(U_a,V_a)_{a\in\Theta}$ of $q$, we first define a query $q_\sigma$ by replacing in $q$, for each $a\in\Theta$, all variables in $U_a$ (if any) with $x_a$. 
%
Then, we define queries $\widehat {q_\sigma}$ and $q_{\sigma}^a$ by taking:
\begin{align*}
\widehat {q_\sigma} & = q_\sigma|_{\{x_a\mid a\in \Theta\}} \\
q_\sigma^a & =q_\sigma|_{\{x_a\}\cup V_a}, \text{ for each $a\in \Theta$}
\end{align*}
We treat $q_\sigma^a$ as a unary CQ with free variable $x_a$ if $U_a\neq\emptyset$, and there is an atom $r(x,y)\in q$ with $x\in U_a$ and $y\in V_a$, and as a Boolean CQ otherwise. Let us finally denote with $\delta_{\Theta}$ the map defined by taking $\delta_\Theta(x_a)=a$, for each $a\in\Theta$.

We call a $\Theta$-split $\sigma=(U_a,V_a)_{a\in\Theta}$ of a subdivision of $q$
\emph{admissible} if the following two conditions are satisfied:
\begin{itemize}

  \item for each $a\in \Theta$, $q_\sigma^a$ has a match to the transitive
  closure of a tree-shaped interpretation, and

  \item every variable $z$ introduced in the subdivision is contained in some $U_a$.
  
\end{itemize}
We have the following characterization of admitting a match to forest-interpretations in terms of splits. 
\begin{restatable}{lemma}{matchsplitCQ}
\label{lem:splitCQ}
    Let $\interp$ be a $\Theta$-forest interpretation with induced tree-shaped interpretations $(\Imc_a)_{a\in \Theta}$. Then a Boolean $q$ admits a match to $\interp$ if and only if there exists an admissible $\Theta$-split $\sigma=(U_a,V_a)_{a\in\Theta}$ of a subdivision $p$ of $q$ such that:
    \begin{enumerate}[label=(\alph*),leftmargin=*]
       
        \item $\delta_\Theta$ is a match of $\widehat {p_\sigma}$ in $\Imc$, and
        
        \item for every $a\in\Theta$, there is a match $\delta_a$ of $p_\sigma^a$ in $\Imc_a^+$ which satisfies $\delta_a(x_a)=a$ in case $U_a\neq \emptyset$.
        
    \end{enumerate}
\end{restatable}

\begin{proof}
For "if", suppose that there exists an admissible $\Theta$-split $\sigma=(U_a,V_a)_{a\in\Theta}$ of a subdivision of $q$, that satisfies Items~(a) and~(b) from the lemma. It is routine to verify that the union of $\delta_U$ defined by taking
\[\delta_U(z)=a\text{ for all $a\in \Theta$, $z\in U_a$}\]
and the matches $\delta_a$, $a\in\Theta$ from Item~(b) is actually a match of $q$ in $\Imc$.

 For "only if", suppose that there exists a match $\delta$ of $q$ in~$\Imc$. We first define a subdivision $q'$ of $q$ and a match $\delta'$ of $q'$ in $\Imc$ as follows. Start with $\delta'=\delta$ and then:  
 \begin{itemize}
 
     \item Replace every atom $r(x,y)\in q$ with $r$ transitive and such that $\delta(x)=a\in\Theta$ and $\delta(y)\in \Delta^{\Imc_b}\setminus\{b\}$ for some $b\neq a$ with $r(x,z),r(z,y)$ for some fresh variable $z$ and set $\delta'(z)=b$.
     
 \end{itemize}
 It should be clear that the resulting $\delta'$ is a match of $q'$ in $\Imc$. Now, we define a $\Theta$-split of $q'$ by taking, for $a\in\Theta$:
 \begin{align*}
     U_a&=\{z\mid \delta'(z)=a\} \\
     V_a&=\{z\mid \delta'(z)\in \Delta^{\Imc_a}\setminus\{a\}\}
 \end{align*}
 It is not difficult to verify that $\sigma=(U_a,V_a)_{a\in\Theta}$ satisfies Conditions~(S1)-(S4) of splits and is admissible. Moreover, Item~(a) is satisfied, and Item~(b) is witnessed by $\delta'$ extended with $\delta'(x_a)=a$ for all $a\in \Theta$.
\end{proof}

There is the following analogue for rooted CQs. 
\begin{restatable}{lemma}{matchsplitRCQ}
\label{lem:splitRCQ}
    Let $\interp$ be a $\Theta$-forest interpretation with induced tree-shaped interpretations $(\Imc_a)_{a\in \Theta}$ and $a_0\in\Theta$, and let $\bar a$ be a tuple from $\Theta$. Then a rooted CQ $q(\bar x)$ admits a match $\delta$ in $\interp$ with $\delta(\bar x)=\bar a$ if and only if there exists an admissible $\Theta$-split $\sigma=(U_a,V_a)_{a\in\Theta}$ of a subdivision $p$ of $q$ such that:
    \begin{enumerate}[label=(\alph*),leftmargin=*]

        \item $x\in U_{\delta(x)}$, for all $x\in \bar x$,
       
        \item $\delta_\Theta$ is a match of $\widehat {p_\sigma}$ in $\Imc$, and
        
        \item for every $a\in\Theta$, there is a match $\delta_a$ of $p_\sigma^a$ in $\Imc_a^+$ which satisfies $\delta_a(x_a)=a$ in case $U_a\neq \emptyset$.
        
    \end{enumerate}
\end{restatable}

We next prove an auxiliary statement that will enable the proof of Theorem~\ref{th:reduction-rooted}. Intuitively, it is a reduction from entailment of rooted UCQs to entailment of the form $\pair{\Tmc}{\tau}\models Q$, as introduced in the main part. 

\begin{proposition} \label{pr:reduction-rooted}
Let $\Kmc=\pair{\tbox}{\abox}$ be a knowledge base, $Q(\bar x)$ be a rooted UCQ with $\bar x=\langle x_1,\ldots,x_n\rangle$, $\bar a=\langle a_1,\ldots,a_n\rangle$ a tuple from $\NI(\Amc)$, and let $\Theta = \NI(\Amc)$. Then $\Kmc\nvDash Q(\bar a)$ if and only if there is a model $\Jmc$ of $\abox$ and $\{ A\sqsubseteq \forall r.B\mid A\sqsubseteq \forall r.B\in \Tmc\}$ that only interprets the concept and role names in $\Kmc$ non-empty, and a family $(Q^1_a)_{a\in\Theta}$ of unary rooted UCQs, satisfying the following:
\begin{enumerate}

    \item For every $q\in Q$, every subdivision $p$ of $q$ and every admissible $\Theta$-split $\sigma=(U_a,V_a)_{a\in\Theta}$ of $p$ such that $x_i\in U_{a_i}$ for each $i\leq n$ and $\delta_{\Theta}$ is a match of $\widehat{p_{\sigma}}$ in $\Jmc$, there is some $a\in\Theta$ such that $U_a\neq\emptyset$ and $p_\sigma^a(x_a)\in Q_a^1$.
    
    \item For every $a\in \Theta$, we have  $\pair{\tbox}{\tp(\Jmc,a)}\not\models Q^1_a$.
     
\end{enumerate}
Moreover, in the ``only if'' direction, the family of rooted UCQs can be chosen so that the cardinality of 
each $Q^1_a$ is exponential in $\|Q\|$, yet the size of CQ in $Q^1_a$ is polynomial (even linear) in $\|Q\|$.
\end{proposition}

\begin{proof}
  For the "if"-direction, let us assume that there is a model $\Jmc$ of $\abox$ and $\{ A\sqsubseteq \forall r.B\mid A\sqsubseteq \forall r.B\in \Tmc\}$ that interprets at most the concept and role names occurring in $\Kmc$ non-empty and a family $(Q^1_a)_{a\in\indiv(\abox)}$ satisfying Points~1 and~2 from the statement. Note that $\Jmc=\Jmc^+$ since $\Jmc\models \{ A\sqsubseteq \forall r.B\mid A\sqsubseteq \forall r.B\in \Tmc\}$.

By Item~2, we have $\pair{\tbox}{\tp(\Jmc,a)}\nvDash Q^1_a$,  for each $a\in\indiv(\abox)$. Hence, by Lemma~\ref{lem:forest}, for each $a\in\indiv(\abox)$, there are tree-shaped interpretations $\interp_a$, rooted in $a$, with $\tp(\Imc_a,a)=\tp(\Jmc,a)$, such that $\interp_a^+\vDash\tbox$ and $\pair{\interp_a^+}{a}\nvDash Q^1_a$. We can assume without loss of generality that the $\Imc_a$ have pairwise disjoint domains.
 
 We define $\interp$ as the transitive closure of the union of the following interpretations:
 \begin{itemize}
     \item the restriction $\Jmc_\Amc$ of $\Jmc$ to domain $\indiv(\Amc)$, 
     \item $\Imc_a$, for $a\in\indiv(\Amc)$.
 \end{itemize}
 Note that it is not a disjoint union since $\Jmc_\Amc$ and $\Imc_a$ share domain element $a$, for each $a\in\indiv(\Amc)$. This is not a problem because of matching types: indeed, by construction, we have
 $\tp(\interp_a,a)=\tp(\Jmc,a)$. We claim that $\interp$ is a model of $\Kmc$ but $\pair{\interp}{\bar a}\not\models Q(\bar x)$. 
 
 For the former, note first that $\Imc$ is a model of $\Amc$ since $\Jmc_\Amc$ is a model of $\Amc$. To show that \Imc is a model of \Tmc, we only have to show something for elements in $\indiv(\Amc)$. However, it is routine to show that all inclusions in \Tmc are satisfied using the fact that each $\Imc_a^+$ is a model of \Tmc and the fact that $\Jmc$ is a model of $\{ A\sqsubseteq \forall r.B\mid A\sqsubseteq \forall r.B\in \Tmc\}$.

 For the latter, $\pair{\Imc}{\bar a}\not\models Q(\vec x)$, let us assume that $\pair{\Imc}{\bar a}\models q(\vec x)$ for some $q\in Q$. Since $\Imc$ is a $\Theta$-forest interpretation with induced $\Imc_a$, $a\in\Theta$, Lemma~\ref{lem:splitRCQ} implies that there exists an admissible $\Theta$-split $\sigma=(U_a,V_a)_{a\in\Theta}$ of a subdivision $p$ of $q$ such that $x\in U_{\delta(x)}$ for all $x\in\bar x$, $\delta_\Theta$ is a match of $\widehat {p_\sigma}$ in $\interp$, and for every $a\in\Theta$, $p_{\sigma}^a$ admits a match $\delta_a$ to $\Imc_a^+$ which additionally satisfies $\delta_a(x_a)=a$ in case $U_a\neq \emptyset$. Item~1 of our assumptions implies that $p_\sigma^c\in Q_c^1$, for some $c$ with $U_c\neq \emptyset$. But then $\pair{\Imc_c^+}{c}\models Q_c^1$, in contradiction to our choice of $\Imc_c$.

\medskip
 For the "only if"-direction, let us assume that $\Kmc\nvDash Q(\bar a)$. By Lemma~\ref{lem:forest}, there exists an $\indiv(\abox)$-forest interpretation $\interp$ such that $\interp\vDash\Kmc$ and $\pair{\Imc}{\bar a}\not\models Q(\vec x)$. Without loss of generality, we can assume that $\Imc$ interprets only concept and role names from $\Kmc$ as non-empty. Let $\Jmc=\Imc$ and $(\interp_a)_{a\in\indiv(\abox)}$ the tree-shaped interpretations induced by $\Imc$ and $\Theta = \indiv(\Amc)$. We show how to construct the family $(Q_a^1)_{a\in \Theta}$ of rooted CQs. 
 
Consider any admissible $\Theta$-split $\sigma=(U_a,V_a)_{a\in\Theta}$ of some subdivision $p$ of some $q\in Q$ such that $x_i\in U_{a_i}$ for $i=1..n$ and $\delta_\Theta$ is a match of $\widehat {p_{\sigma}}$ in $\interp$. Since $q$ does not have a match in $\Imc$, Lemma~\ref{lem:splitRCQ} tells us that for some $a\in\indiv(\abox)$ with $U_a\neq \emptyset$, $\pair{\Imc_a^+}{a}\not\models p_\sigma^a$. We add this $p_\sigma^a$ to $Q_a^1$.

Clearly, the constructed set $Q_a^1$ satisfies the required size conditions and Item~1 of the lemma. It remains to show that Item~2 of the lemma is also satisfied, that is, $\pair{\Tmc}{\tp(\Jmc,a)}\not\models Q^1_a$. But this is easy: $\Imc_a^+$ is a model of $\Tmc$, $a$ satisfies $\tp(\Imc_a,a)=\tp(\Jmc,a)$, and $\pair{\Imc_a^+}{a}\not\models q$ for all queries put into $Q_a^1$, simply by construction. 
\end{proof}

We need one more ingredient to finally prove Theorem~\ref{th:reduction-rooted}. The following auxiliary lemma states that we can transform a unary CQ into a union of TQs which are equivalent over tree-shaped interpretations.

\begin{lemma}\label{lem:treeification}
Let $q(x)$ be a unary CQ with $n$ variables. Then, one can compute in time exponential in the size of $q$, a UTQ $Q(x)$ such that
\begin{enumerate}[leftmargin=*]

    \item  $\pair{\Imc}{a}\models q$ iff $\pair{\Imc}{a}\models Q$, for all transitive-tree interpretations $\Imc$ rooted at $a$, and 
    
    \item each TQ in $Q$ has at most $n$ variables.
    
\end{enumerate}
\end{lemma}

\begin{proof}
Let $q(x)$ be a rooted unary CQ with $n$ variables and mentioning concept names $\mathsf{CN}$ and role names $\mathsf{RN}$. 

For the construction, we consider the class $\mathbf I$ of tree interpretations $\Imc$ with domain contained in $\{1,\ldots,n\}$ and root $1$ that interpret only $\mathsf{CN}$ and $\mathsf{RN}$ as possibly non-empty. Each 
such $\Imc$ can be viewed as a TQ $q_\Imc(x_1)$ as follows: 
\begin{itemize}

\item $q_\Imc$ contains atoms $r(x_i,x_j)$ whenever $(i,j)\in r^\Imc$ and $r\in\mathsf{RN}$, and atoms $A(x_i)$ whenever $i\in A^\Imc$ and $A\in\mathsf{CN}$. 

\item The answer variable of $q_\Imc$ is $x_1$.

\end{itemize}

We define $Q$ as the set of all $q_\Imc$ with $\Imc\in\mathbf I$ such that
$\pair{\Imc^+}{1}\models q(x)$. Clearly, $Q$ can be computed in exponential time
since $\mathbf I$ can be constructed in exponential time. We claim that $Q$
satisfies Points~1 and~2 from the lemma. Item~2 is satisfied by construction. 

For the "if"-direction of Item~1, suppose that $\pair{\Imc}{a}\models Q$ for some transitive-tree interpretation $\Imc$. Then $\pair{\Imc}{a}\models q_\Jmc(x_1)$ for some TQ $q_\Jmc\in Q$.
Since $q_\Jmc\in Q$, there is a match $\delta'$ of $q$ to $\Jmc^+$ with $\delta'(x)=1$. Then the composition of $\delta$ and $\delta'$, that is, the map $\hat\delta$ defined by taking 
\[\hat\delta(y)=\delta(x_{\delta'(y)})\text{ for all $y\in\var(q)$}\]
is a match of $q$ in $\Imc$ with $\hat\delta(x)=a$. Hence $\pair{\Imc}{a}\models q$.

For the "only if"-direction of Item~1, suppose that $\pair{\Imc}{a}\models q$ for some transitive-tree interpretation $\Imc$, that is, there is a match $\delta$ of $q$ into $\Imc$ with $\delta(x)=a$.
Define an interpretation $\Jmc$ with domain $\Delta^\Jmc=\{\delta(y)\mid y\in\var(q)\}$ as follows: 
\begin{itemize}
    \item $d\in A^\Jmc$ if $d\in A^{\Imc}$, for all $d\in \Delta^\Jmc$ and $A\in \mathsf{CN}$,
    
    \item $(d,e)\in r^\Jmc$ if $(d,e)\in r^{\Imc}$, for all $d,e\in \Delta^\Jmc$ and non-transitive $r\in \mathsf{CN}$,
    
    \item $(d,e)\in t^\Jmc$ if $(d,e)\in t^{\Imc}$ and there is no $f\in \Delta^\Jmc$ with $(d,f),(f,e)\in t^{\Imc}$, for all $d,e\in \Delta^\Jmc$ and transitive $t\in \mathsf{CN}$.
      
\end{itemize}
We can rename the domain elements of $\Jmc$ such that they are from $\{1,\ldots,n\}$ and such that $1$ is the root. In particular, $a$ is renamed to $1$. Clearly, $\Jmc\in \mathbf I$. By definition, $q$ has a match to $\Jmc^+$, hence $q_\Jmc\in Q$. Since (up to renaming) $\Jmc$ is a sub-interpretation of $\Imc$, $q_\Jmc$ has a match $\delta$ in $\Imc$ with $\delta(x_1)=a$. Thus, $\pair{\Imc}{a}\models Q$ as required.
\end{proof}

\reductionrooted*

\begin{proof}
The non-deterministic algorithm starts with guessing an interpretation $\Imc_0$ with domain $\NI(\Amc)$ which interprets only concept and role names occurring in $\pair{\Tmc}{\Amc}$ non-empty. Note that there are only exponentially many such interpretations. Set $\tau_a=\tp(\Imc_0,a)$ for all $a\in \indiv(\Amc)$, and let $\Jmc=\Imc_0^+$ be the transitive closure of $\Imc_0$.

 The algorithm checks whether $\Jmc$ is a model of $\Amc$ and $\{ A\sqsubseteq \forall r.B\mid A\sqsubseteq \forall r.B\in \Tmc\}$. If this is not the case, it rejects.

  Let $\Theta =\indiv(\Amc)$, $\bar x=\langle x_1,\ldots,x_n\rangle$, and $\bar a=\langle a_1,\ldots,a_n\rangle$. The algorithm computes all admissible $\Theta$-splits $\sigma=(U_a,V_a)_{a\in \Theta}$ of subdivisions $p$ of some $q\in Q$. Note that there are only exponentially many $\Theta$-splits and that admissibility can be checked in polynomial time, by Lemma~\ref{lem:matchtotreeshaped}. If $x_i\in U_{a_i}$ for all $i=1..n$ and $\delta_\Theta$ is a match of $\widehat{q_\sigma}$ in $\Jmc$, then the algorithm non-deterministically does the following:
  \begin{itemize}
  
      \item[$(\ast)$] pick $a\in\Theta$ with $U_a\neq\emptyset$, and compute, via Lemma~\ref{lem:treeification}, a unary UTQ $P_\sigma^a(x)$ equivalent to $p_\sigma^a(x)$ (in the sense of Item~1 of Lemma~\ref{lem:treeification}). Then construct $\widehat P_\sigma^a(x)$ from $P_\sigma^a(x)$ by picking (non-deterministically) for each $p(x)\in P_\sigma^a(x)$ one atom $r(x,y)$ and dropping all subqueries starting in an atom $s(x,y')\neq r(x,y)$, including that atom. Then 
      add (all disjuncts of) $\widehat P_\sigma^a$ to $Q_a^1$.
      
  \end{itemize}
  After treating all subdivisions $p$ of a query from $q\in Q$ and all their $\Theta$-splits in this way, the algorithm outputs $\tau_a$ and $Q^1_a$, for all $a\in\NI(\Amc)$.

  It is routine to verify that the algorithm runs in (non-deterministic) exponential time. Moreover, based on Proposition~\ref{pr:reduction-rooted} and Lemma~\ref{lem:treeification}, it can be shown that it behaves
  as claimed in Theorem~\ref{th:reduction-rooted}. We give some details.

  Suppose $\pair{\Tmc}{\Amc}\not\models Q(\bar a)$, and let $\Theta=\indiv(\Amc)$.
  By Proposition~\ref{pr:reduction-rooted}, there is a model $\Jmc$ of \Amc and
$\{ A\sqsubseteq \forall r.B\mid A\sqsubseteq \forall r.B\in \Tmc\}$ that interprets at most concept and role names from $\pair{\Tmc}{\Amc}$ non-empty and a
family of rooted UCQs $(Q_a^1)_{a\in\Theta}$ satisfying Points~1 and~2 from Proposition~\ref{pr:reduction-rooted}. We show that the algorithm has a run which outputs families $(\tau_a)_{a\in\indiv(\Amc)},(\widehat Q_a^1)_{a\in\indiv(\Amc)}$ such that $\pair{\Tmc}{\tau_a}\not\models \widehat Q_a^1$, for all $a\in\indiv(\Amc)$.
Let $\Jmc_0$ be the restriction of $\Jmc$ to $\indiv(\Amc)$. This $\Jmc_0$ can be guessed by the algorithm in the first step, and $\tau_a$ is just taken as $\tau_a=\tp(\Jmc_0,a)$, for all $a$.

Let $p$ be a subdivision of some $q\in Q$ and let $\sigma=(U_a,V_a)_{a\in\Theta}$ be an arbitrary admissible $\Theta$-split of $p$ such that $x_i\in U_{a_i}$ for $i=1..n$ and $\delta_\Theta$ is a match of $\widehat{p_\sigma}$ in $\Jmc_0$. By Item~1 of Proposition~\ref{pr:reduction-rooted}, there is some $a$ such that $U_a\neq \emptyset$ and $p_\sigma^a\in Q_a^1$. We let the non-deterministic algorithm pick this $p_\sigma^a$, and add the UTQ $P_\sigma^a$ equivalent to $p_\sigma^a$ to $\widehat Q_a^1$ as described in Step~$(\ast)$.

We claim that the computed $\tau_a,\widehat Q_a^1$ satisfy $\pair{\Tmc}{\tau_a}\not\models \widehat Q_a^1$, for all $a\in\Theta$. Suppose otherwise, that is, there is a model $\Imc$ and element $a$ with $\tp(\Imc,a)=\tau_a$ such that $\pair{\Imc}{a}\models \widehat Q_a^1$. By Lemma~\ref{lem:forest}, we can assume that $\Imc=\Imc_0^+$ for some tree-shaped $\Imc_0$. Then $\pair{\Imc}{a}\models P^a_\sigma$ for some $P_\sigma^a$, added in Step~$(\ast)$. By Lemma~\ref{lem:treeification}, we have
$\pair{\Imc}{a}\models p^a_\sigma$. But since $p^a_\sigma\in Q_a^1$, this is in contradiction to Item~2 of Proposition~\ref{pr:reduction-rooted}, that is, $\pair{\Tmc}{\tp(\Jmc_0,a)}\not\models Q_a^1$.

The other direction is shown using similar arguments. 
\end{proof}

We can now state the promised characterization of non-entailment. Remember that, if $q$ is a CQ that has a match in a \new{transitive-tree interpretation}, then $q_\approx$ can be decomposed into a conjunction of PTQs which is equivalent to $q$ over the class of \new{transitive-tree interpretations} (see Lemma~\ref{lem:to_ptq}).

\begin{proposition} \label{pr:reduction-single}
Let $\Kmc=\pair{\tbox}{\abox}$ be a knowledge base, and $Q$ be a Boolean UCQ that mentions a single transitive role name $t$. Let $\Theta=\indiv(\Amc)$. Then $\Kmc\nvDash Q$ if and only if there is a model $\Jmc$ of $\abox$ and $\{ A\sqsubseteq \forall r.B\mid A\sqsubseteq \forall r.B\in \Tmc\}$ that interprets at most the concept and role names occurring in $\Kmc$ non-empty, a family $(Q^0_a)_{a\in\Theta}$ of Boolean UPTQs, and a family $(Q^1_a)_{a\in\Theta}$ of unary UPTQs, satisfying the following:
\begin{enumerate}

    \item For every $q\in Q$ and every admissible $\Theta$-split $\sigma=(U_a,V_a)_{a\in \Theta}$ of a subdivision $p$ of $q$ such that $\delta_\Theta$ is a match of $\widehat{p_{\sigma}}$ to $\Jmc$, there is some $a\in\Theta$ such that one of the following holds:
    \begin{itemize}
        
        \item $U_a\neq \emptyset$ and the connected component of $p_{\sigma}^a(x_a)$ containing $x_a$ is in $Q^1_a$; or
        
        \item $V_a\neq\emptyset$ and some Boolean subquery of  $p_{\sigma}^a$ is in $Q^0_a$.
        
    \end{itemize}
    
    \item For every $a\in \Theta$, we have $\pair{\tbox}{\tp(\Jmc,a)}\not\models Q^0_a\vee Q^1_a$.
     
\end{enumerate}
Moreover, in the ``only if'' direction, the UPTQs can be chosen so that, for every $a\in\indiv(\abox)$, the size of $Q^0_a$ is linear in $\|Q\|$, the size of
$Q^1_a$ is exponential in $\|Q\|$, yet the number of subPTQs of PTQs in $Q^1_a$ is polynomial in $\|Q\|$.

\end{proposition}

\begin{proof}
  For the "if"-direction, let us assume that there is a model $\Jmc$ of $\abox$ and $\{ A\sqsubseteq \forall r.B\mid A\sqsubseteq \forall r.B\in \Tmc\}$ that interprets at most the concept and role names occurring in $\Kmc$ non-empty and a family $(Q^1_a)_{a\in\indiv(\abox)}$ satisfying Points~1 and~2 from the statement.

By Item~2, we have $\pair{\tbox}{\tp(\Jmc,a)}\nvDash Q^0_a\vee Q^1_a$, for each $a\in\Theta$. Hence, by Lemma~\ref{lem:forest}, for each $a\in\Theta$, there are tree-shaped interpretations $\interp_a$, rooted in $a$, with $\tp(\Imc_a,a)=\tp(\Jmc,a)$, and such that $\interp_a^+\vDash\tbox$, $\Imc_a^+\not\models Q_a^0$, and $\pair{\interp_a^+}{a}\nvDash Q^1_a$. We can assume without loss of generality that the $\Imc_a$ have pairwise disjoint domains.
 
 We define $\interp$ as the transitive closure of the union of the following interpretations:
 \begin{itemize}
     \item the restriction $\Jmc_\Amc$ of $\Jmc$ to domain $\indiv(\Amc)$, 
     \item $\Imc_a$, for $a\in\Theta$.
 \end{itemize}
 Note that it is not a disjoint union since $\Jmc_\Amc$ and $\Imc_a$ share domain element $a$, for each $a\in\indiv(\Amc)$. This is not a problem because of matching types: indeed, by construction, we have
 $\tp(\interp_a,a)=\tp(\Jmc,a)$. We claim that $\interp$ is a model of $\Kmc$ but $\Imc\not\models Q$. 
 
 For the former, note first that $\Imc$ is a model of $\Amc$ since $\Jmc$ is a model of $\Amc$. To show that \Imc is a model of \Tmc, we only have to show something for elements in $\indiv(\Amc)$. However, it is routine to show that all inclusions in \Tmc are satisfied using the fact that each $\Imc_a^+$ is a model of \Tmc and the fact that $\Jmc$ is a model of $\{ A\sqsubseteq \forall r.B\mid A\sqsubseteq \forall r.B\in \Tmc\}$.

 For the latter, $\Imc\not\models Q$, let us assume that $\Imc\models q$ for some $q\in Q$. Since $\Imc$ is a $\Theta$-forest interpretation with induced $\Imc_a$, Lemma~\ref{lem:splitCQ} implies that there exists an admissible $\Theta$-split $\sigma=(U_a,V_a)_{a\in\Theta}$ of a subdivision $p$ of $q$ such that $\delta_\Theta$ is a match of $\widehat {p_\sigma}$ in $\interp$, and for every $a\in\Theta$, $p_{\sigma}^a$ admits a match $\delta_a$ to $\Imc_a^+$ which additionally satisfies $\delta_a(x_a)=a$ in case $U_a\neq \emptyset$. Item~1 of our assumptions implies that either the connected component of $p_\sigma ^c(x_c)$ is contained $Q_c^1$, for some $c$ with $U_c\neq \emptyset$, some Boolean subquery of $p_\sigma^c$ is contained in $Q_c^0$ for some $c$ with $V_c\neq\emptyset$.
 
 But then either $\pair{\Imc_c^+}{c}\models Q_c^1$ or $\Imc_c^+\models Q_c^0$, in contradiction to our choice of $\Imc_c$.

\medskip
 For the "only if"-direction, let us assume that $\Kmc\nvDash Q$. By Lemma~\ref{lem:forest}, there exists an $\indiv(\abox)$-forest interpretation $\interp$ such that $\interp\vDash\Kmc$ and $\Imc\not\models Q$. Without loss of generality, we can assume that $\Imc$ interprets only concept and role names from $\Kmc$ as non-empty. Let $\Jmc=\Imc$ and $(\interp_a)_{a\in\indiv(\abox)}$ the tree-shaped interpretations induced by $\Imc$ and $\Theta = \indiv(\Amc)$. We show how to construct the families $(Q_a^0)_{a\in\Theta}$ and $(Q_a^1)_{a\in \Theta}$ of UPTQs. 
 
Consider any admissible $\Theta$-split $\sigma=(U_a,V_a)_{a\in\Theta}$ of some subdivision $p$ of some $q\in Q$ such that $\delta_\Theta$ is a match of $\widehat {p_{\sigma}}$ in $\interp$. Since $q$ does not have a match in $\Imc$, Lemma~\ref{lem:splitCQ} tells us that for some $a\in\indiv(\abox)$ with $U_a\neq \emptyset$, $\pair{\Imc_a^+}{a}\not\models p_\sigma^a$. This means that either the connected component $p'$ of $p_\sigma^a$ containing $x_a$ satisfies 
$\pair{\Imc_a^+}{a}\not\models p'$ or some Boolean subquery $p''$ of $p_\sigma^a$ satisfies $\Imc^+_a\not\models p''$. In the former case, we add the PTQs obtained from $p'$ via Item~2 of Lemma~\ref{lem:to_ptq} to $Q_a^1$, and in the latter case, the PTQ obtained from $p''$ via Item~1 of Lemma~\ref{lem:to_ptq} to $Q_a^0$.

Clearly, the constructed set $Q_a^1$ satisfies Item~1 of the lemma. We show that Item~2 of the lemma is satisfied as well, that is, $\pair{\Tmc}{\tp(\Jmc,a)}\not\models Q_a ^0 \vee Q^1_a$, for all $a\in\Theta$. But this is easy: $\Imc_a^+$ is a model of $\Tmc$, $a$ satisfies $\tp(\Imc_a,a)=\tp(\Jmc,a)$, and $\pair{\Imc_a^+}{a}\not\models q$ for all queries put into $Q_a^1$, and $\Imc_a ^+\not\models q$ for all queries put into $Q_a^0$, simply by construction. 

\smallskip
It remains to argue that the claimed size restrictions are satisfied as well. We begin with analyzing the number of PTQs in $Q^0_a$. By the definition of splits, every PTQ added to $Q^0_a$ is equivalent to a Boolean subquery of some CQ in $Q$, of which there clearly are only linearly many. 

Let us now analyze the number of subPTQs of queries in $Q^1_a$. For this purpose, let us fix some $q\in Q$ and recall the equivalence relation $\approx$ from Definition~\ref{def:queryequivalence}. Let $\sigma=(U_a,V_a)_{a\in \Theta}$ be an admissible $\Theta$-split of some subdivision $p$ of $q$.
We make a few observations.

\medskip\noindent\textit{Claim~1.} For every $a\in \Theta$ and $x\approx y$: if $x\in V_a$, then $y\in V_a$. 

\medskip\noindent\textit{Proof of Claim~1.} This follows by a routine induction on the definition of $\approx$ (when viewed as a saturation procedure).
%
\hfill$\dashv$

\medskip This means that in any admissible split the $V_a$ are always finite unions of equivalence classes of $\approx$ (and clearly $\approx$ does not have more than $|\var(q)|$ equivalence classes). 
Let $q_1(x'),\ldots,q_m(x')$ be the decomposition of $q_\approx$ into unary PTQs that is computed in the proof of Lemma~\ref{lem:to_ptq} when applied to a unary CQ $q^a_\sigma(x)$ which has a match in a \new{transitive-tree interpretation}. Recall that $x'$ is an equivalence class of $\approx$ (actually, the class of $x$). By construction the root cluster of each $q_i(x')$ contains an atom $r(x',[y])$ for some atom $r(x_0,y)\in q$ with $x_0\in U_a$ and $y\in V_a$. Since $V_a$ is closed under successors (see Condition~(S3)) of splits, this atom $r(x_0,y)$ together with $\approx$ determines $q_i(x')$. Since there are only linearly many binary atoms in $q$, the linear upper bound on the number of subPTQs follows.
\end{proof}

With Proposition~\ref{pr:reduction-single} at hand, it is not difficult to prove Theorem~\ref{th:reduction-single}, which we restate here for the reader's convenience. The proof is essentially along the lines of the proof of Theorem~\ref{th:reduction-rooted}.

\reductionone*

\begin{proof}
The non-deterministic algorithm starts with guessing an interpretation $\Imc_0$ with domain $\NI(\Amc)$ which interprets only concept and role names occurring in $\pair{\Tmc}{\Amc}$ non-empty. Note that there are only exponentially many such interpretations. Set $\tau_a=\tp(\Imc_0,a)$ for all $a\in \indiv(\Amc)$, and let $\Jmc=\Imc_0^+$ be the transitive closure of $\Imc_0$.

 The algorithm checks whether $\Jmc$ is a model of $\Amc$ and $\{ A\sqsubseteq \forall r.B\mid A\sqsubseteq \forall r.B\in \Tmc\}$. If this is not the case, it rejects.

  Let $\Theta =\indiv(\Amc)$. The algorithm computes all admissible $\Theta$-splits $\sigma=(U_a,V_a)_{a\in \Theta}$ of subdivisions $p$ of some $q\in Q$. Note that there are only exponentially many $\Theta$-splits and that admissibility can be checked in polynomial time, by Lemma~\ref{lem:matchtotreeshaped}. If $\delta_\Theta$ is a match of $\widehat{q_\sigma}$ in $\Jmc$, then the algorithm non-deterministically does the following:
  \begin{itemize}
  
      \item[$(\ast)$] pick $a\in\Theta$ with $U_a\cup V_a\neq \emptyset$ and
      do one of the following: 
      \begin{itemize}
          \item verify that $U_a\neq \emptyset$ and let $q_1(x'),\ldots,q_k(x')$ be the sequence of unary PTQs obtained from applying Lemma~\ref{lem:to_ptq} to the connected component of $p_\sigma^a$ containing $x_a$. Then pick some $q_i(x')$ and put it into $Q_a^1$, or
          

          \item verify that $V_a\neq\emptyset$ and (non-deterministically) pick a Boolean subquery $p'$ of $p_\sigma^a$ and put the PTQ equivalent to $p'$ (obtained via Lemma~\ref{lem:to_ptq}) into $Q_a^0$.
          
      \end{itemize}
      
  \end{itemize}
   After treating all $q\in Q$ and all splits of subdivisions of $q$ in this way, the algorithm outputs $\tau_a,Q^0_a,Q^1_a$ for all $a\in\NC(\Amc)$.

   We can now verify (along the lines of the proof of Theorem~\ref{th:reduction-rooted}) that the algorithm behaves as claimed in the theorem.
\end{proof}

\subsection{Proof of Theorem~\ref{th:reductionbasic}}
\label{ap:second_reduction}

\new{In this section of the appendix, we show a proof of Theorem~\ref{th:reductionbasic}, highlighting the correspondence between our entailment problem and the existence of mosaics:}

\reductionbasic*

\begin{proof} 
Assume first that $\pair{\Tmc}{\tau}\not\models Q^0\vee Q^1$ for some TBox \Tmc, a set $\tau\subseteq \NC(\Tmc)$, a Boolean UPTQ $Q^0$, and a unary UPTQ $Q^1$. It follows from the same arguments used to prove Lemma~\ref{lem:forest} that there is a transitive-tree interpretation $\Imc$ rooted at some $d^*$ such that $\Imc\models \Tmc$, $\text{tp}(\Imc,d^*)=\tau$, $\Imc\not\models Q^0$, and $\langle\Imc, 
d^*\rangle\nvDash Q^1$. We can further assume without loss of generality that $\Imc$ interprets only concept names occurring in $\Tmc$ as non-empty. We extend $\Imc$ by interpreting the fresh concept names $A_{p(x)}$ by taking:
\[A_{p(x)}^\Imc = \{d\in\Delta^\Imc\mid \pair{\Imc}{d}\models p(x)\}.\]
Let $\Sigma_C$ be the set of all concept names interpreted non-empty in $\Imc$ after this extension. 

To every domain element $d\in\Delta^\Imc$ and every role name $r\in \NR(\Tmc)$, we associate a tile $\langle \Imc_{d,r},d,r\rangle$ where $\Imc_{d,r}$ is the restriction of $\Imc$ to domain $\{d\}\cup\{e \in \Delta^\Imc \mid (d,e)\in r^\Imc\}$ and signature $\Sigma_C\cup\{r\}$. It is routine to verify that:

\medskip\noindent\textbf{Claim 1.} $\langle\Imc_{d,r},d,r\rangle$ is a tile, for each $d\in \Delta^\Imc$, $r\in \NR(\Tmc)$.

\medskip 
Let $\Mos$ be the set of all tiles defined in this way. It is routine to verify: 

\medskip\noindent\textbf{Claim 2.} $\Mos$ is a mosaic for \Tmc and $\tau$, and against $Q^0$ and $Q^1$. 

\medskip
Conversely, take a mosaic $\Mos$ for $\Tmc$ and $\tau$, and against $Q^0$ and $Q^1$. 
We assume without loss of generality that the tiles in $\Mos$ have pairwise disjoint domains.
We construct, in an inductive way, a model $\Imc$ witnessing $\pair{\Tmc}{\tau}\not\models Q^0\vee Q^1$. The idea is to patch the tiles in an appropriate fashion. Throughout the process, we store provenance information for each introduced domain element. Formally, we will associate to every $d\in\Delta^\interp$ a tile $\langle \Imc_{d},e_d,r_d\rangle$ and an element $d'\in \Delta^{\Imc_d}$ with $\text{tp}(\Imc,d)=\tp(\Imc_d,d')$. This is convenient as we have to rename elements throughout the construction.

For the inductive base, 
let $\langle\Imc_r,d_r,r\rangle\in \Mos$, for 
$r\in \NR(\Tmc)$ be a family of tiles witnessing Condition~3 of $\Mos$ being a mosaic. We rename elements in the $\Imc_r$ in a way that they have disjoint domains except that $d_r=d^*$, for each $r$. Then $\Imc$ is the union of all $\Imc_r,r\in\NR(\Tmc)$. For the provenance, we set,
for each $d\in \Delta^{\Imc_r}\setminus\{d_0\}$ and $r\in \NR(\Tmc)$, 
\[\langle\Imc_d,e_d,r_d\rangle:=\langle\Imc_r,d_r,r\rangle\]
and let $d'$ be the name of $d$ before the renaming process. 
(We do not need provenance information for the root element $d^*$.)

In the inductive step, we apply the following rule in a fair and exhaustive way:
\begin{itemize}

    \item[$(\dagger)$] Take any element $d\in \Delta^\Imc$ with associated tile
    $\langle \Imc_d,e_d,r_d\rangle\in\Mos$ and element $d'\in \Delta^{\Imc_d}$. Moreover, let $s$ be any role name
    such that $s=r_d$ only if both $r_d$ is non-transitive and $e_d\neq d'$. By
    Condition~4 of Definition~\ref{def:mosaic}, there is a tile
    $\langle\Jmc,e_0,s\rangle\in\Mos$ with
    $\text{tp}(\Imc,d)=\text{tp}(\Imc_d,d')=\text{tp}(\Jmc,e_0)$. Let $\Jmc'$ be obtained from
    $\Jmc$ by renaming $e_0$ to $d$ and all other elements in a way that $\Imc$
    and $\Jmc$ share only domain element $d$. Then take the union of $\Imc$ with
    $\Jmc'$. Moreover, for every element $e\in \Delta^{\Jmc'}$ with $e\neq d$, 
    set the provenance $\langle \Imc_{e},e_0,r_e\rangle:=\langle \Jmc,e_0,s\rangle$,
    and define $e'$ to be $e$ before the renaming.
    
\end{itemize} 
 
Let $\Imc$ be the interpretation obtained in the limit of the
above construction. It is not difficult to see that $\text{tp}(\Imc,d^*)\cap
\NC(\Tmc)=\tau$. Moreover, using Condition~4 of mosaics, we can verify that $\Imc\models \Tmc$. It remains to show that $\Imc\not\models Q^0$ and $\pair{\Imc}{d^*}\nvDash Q^1$. To this end, establish the following auxiliary claim.

\medskip\noindent\textbf{Claim 3.} For every $d\in \Delta^\Imc$ and auxiliary concept name $A_{p(x)}$: $d\notin A_{p(x)}^\Imc$ implies $\pair{\Imc}{d}\not\models p(x)$.

\medskip\noindent\textit{Proof of Claim~3.} The proof of the claim is by induction on the canonical cluster tree for the PTQ $p(x)$. 
In the inductive base, $p(x)$ has a single $r$-cluster for some role name $r$. Suppose that $\pair{\Imc}{d}\models p(x)$. But since $p(x)$ is a single cluster, we have that $\pair{\Jmc}{d'}\models p(x)$ for some tile $\langle\Jmc,d',r\rangle$ such that $d'\in \Delta^\Jmc$ is a copy of $d$. But since $d'$ is a copy of $d$, we have $d'\notin A_{p(x)}^{\Jmc}$ and thus $\pair{\Jmc}{d'}\not\models p(x)$ by Condition~2 of mosaics, a contradiction.

In the inductive step, $p(x)$ has several clusters. Let $C$ be the root cluster of $p(x)$ that contains $x$. Suppose it is an $r$-cluster.
Consider $p_{C}(x)$ and let $p_{C}'(x)$ be obtained from $p_{C}(x)$ by dropping all auxiliary atoms. Clearly, $p'_{C}(x)$ is the sub-query of $p$ induced by the variables in cluster $C$. 
Suppose that $\pair{\Imc}{d}\models p(x)$, and let this be witnessed by match $\delta$. Up to renaming, we can assume that $\delta$ is a match of 
$p'_{C}(x)$ in some tile $\langle \Jmc,d_0,r\rangle$ that was involved in the construction $d$ in $\Imc$, that is, $\pair{\Jmc}{d'}\models p'_{C}(x)$ such that $d'\in \Delta^\Jmc$ is a copy of $d$. But since, by assumption, $d'\notin A_{p(x)}^\Jmc$, $\delta$ cannot be a match of $p_{C}(x)$ in $\Jmc$, due to Condition~2 of mosaics. Hence, there is some auxiliary atom $A_{q(z)}(z)\in p_{C}(x)$ such that $\delta(z)\notin A_{q(z)}^\Jmc$, hence $\delta(z)\notin A_{q(z)}^\Imc$. Induction yields $\pair{\Imc}{\delta(z)}\not\models q(z)$, but by assumption $\delta$ witnesses exactly the opposite, a contradiction. 

This finishes the proof of Claim~3.\hfill$\dashv$

\medskip
Relying on arguments similar to those in the proof of Claim~3, we can use Conditions~1 and~3 of mosaics to show that $\Imc\not\models Q^0$ and $\pair{\Imc}{d^*}\not\models Q^1$, respectively. 
This completes the proof of the theorem. 
\end{proof}

\subsection{Proofs of Tile and Mosaic Size Bounds}
\label{sb:sizebounds}

This last section of the appendix is devoted to the proofs of the crucial Lemmas~\ref{cor:mosaic_exp_tile}-\ref{lem:bounded-tiles-rooted}.

\corboundedsize*

\begin{proof}
We show that every tile can be replaced with a finite one of appropriate size, ensuring that the resulting set remains a mosaic for $\Tmc$ and $\tau$, and against $Q^0$ and $Q^1$.

First, let $\langle\interp, d_0, r\rangle\in\Mos$, with $r$ a non-transitive role name. Then we just have to pick a witnessing $r$-successor of $d_0$ in $\Imc$ for each $B \in \NC(\tbox)$ such that there is a CI $A \sqsubseteq \exists r.~B$ in $\tbox$ with $d_0 \in A^\interp$. Let $\Jmc$ be the interpretation of size at most $|\NC(\Tmc)| + 1$ obtained by restricting $\Imc$ to $d_0$ and the picked successors. It is routine to check that $\langle\Jmc, d_0, r\rangle$ is a tile for $\Tmc$ that can replace $\langle\Imc, d_0, r\rangle$ in $\Mos$.

Now, let us consider a tile $\langle\interp, d_0, t\rangle$ for a transitive role name $t$. We use Theorem~\ref{th:acyclic_transitive} to produce a bounded-size tile that can replace $\langle\Imc, d_0, t\rangle$ in $\Mos$ without affecting Conditions 1--4 in Definition~\ref{def:mosaic}.

Towards this end, we encode properties necessary to maintain Conditions 1--3 in a Boolean UPTQ $P$.
To facilitate this, we first introduce additional concept names to $\Imc$. Consider a subPTQ $p(x)$ of a PTQ from $Q^0$ or  $Q^1$, such that the root cluster $C$ of $p$ that contains $x$ is a $t$-cluster. Let $A_{\lnot p(x)}$ be a fresh concept name and define \[A_{\lnot p(x)}^\Imc = \big\{d \in \Delta^\Imc \mid \pair{\Imc}{d}\not\models p_{C}(x)\big\}\,.\] 
Let $P$ collect all Boolean PTQs $p_{C} \cup \{ A_{\lnot p(x)}(x)\}$ where $p$ and $C$ are as above, as well as all $q_C$ with $q\in Q^0$ and $C$ a root $t$-cluster of $q$. 
By construction, $P$ is an acyclic Boolean UCQ over role name $t$, \new{and $\Imc\not\models P$.} 

Applying Theorem~\ref{th:acyclic_transitive} to the role name $t$, the TBox $\Tmc_t$, the Boolean UPTQ $P$, and \new{the unraveling $\Imc'$ of $\Imc$ from $d_0$}, we obtain an interpretation $\Jmc$ of suitable size, such that the tile
$\langle\Jmc, d_0, t\rangle$ can be used to replace $\langle\Imc, d_0, t\rangle$ in $\Mos$.
\new{Indeed, Condition 4 is preserved because, on one hand, 
\[\tp(\Jmc, d_0) = \tp(\Imc', d_0) = \tp(\Imc, d_0)\,,\] 
and on the other hand,  for each $d \in \Delta^\Jmc$ we  have $d \in \Delta^{\Imc'}$, so by the properties of unraveling there is $d'\in\Delta^\Imc$ such that  \[\tp(\Jmc, d)=\tp(\Imc', d) =\tp(\Imc,d')\,.\]}
Conditions 1--3 are preserved because  $\Jmc\not\models P$. For Condition 1, this is immediate. For Condition 2, consider  $d\in \Delta^\Jmc$ such that $d\notin A_{p(x)}^\Jmc$. \new{There is $d'\in\Delta^{\Imc}$ such that \[\tp(\Jmc,d) = \tp(\Imc,d')\,.\] In particular, $d'\notin A_{p(x)}^\Imc$.} Consequently, $\pair{\Imc}{\new{a'}}\not\models p_{C}(x)$ where $C$ is the root $t$-cluster of $p$ that contains $x$. Hence, $\new{a'}\in A_{\lnot p(x)}^\Imc$. Using $\tp(\Jmc, d)=\tp(\interp, \new{d'})$ again, we get $d\in A_{\lnot p(x)}^\Jmc$. Because $\Jmc \not \models p_{C} \cup \{A_{\lnot p(x)}\}$, we conclude that $\pair{\Jmc}{d} \not \models p_{C}(x)$. For Condition~3 the argument is similar\new{, except that we directly have $\tp(\Jmc, d_0) = \tp(\Imc, d_0)$ and do not need the proxy element $d'$}. Hence, we can indeed replace  $\langle\Imc, d_0, t\rangle$ with $\langle\Jmc, d_0, t\rangle$.

Replacing each tile in $\Mos$ like this, we end up with a mosaic for $\Tmc$ and $\tau$ against $Q^0$ and $Q^1$ in which each tile satisfies the desired size bound. 
\end{proof}

Next, we prove Lemma~\ref{lem:bounded-tiles-single}, stating that the mosaics can be assumed to have exponentially many tiles:

\lemboundedtilesingle*

\begin{proof}
    Let us assume that there exists a mosaic $\Mos$ for $\tbox$ and $\tau$ and against $Q^0$ and $Q^1$.
    We define $\Mos'$ as the subset of $\Mos$ consisting of the following tiles:
    \begin{itemize}
        \item $\Mos'$ contains the tiles $(\langle \interp_r, d_r, r\rangle)_{r\in\NC(\tbox)}$ witnessing Condition~\ref{it:initial} of Definition~\ref{def:mosaic}
        \item for each type $\tau'$ and each role name $r\in\NR(\tbox)$ such that there exists a tile $\langle\Jmc, d_0, r\rangle$ with $\tp^{\Jmc}(d_0)=\tau'$, $\Mos'$ contains exactly one of such.
    \end{itemize}

    The elements of tiles satisfy concepts that either occur in $\tbox$, or are of the shape $A_p(x)$, where $p(x)$ is a subPTQ of a PTQ in $Q^0$ or in $Q^1$. Therefore, the obtained set $\Mos'$ has indeed a number of tiles that is bounded by $|\NR(\tbox)| + 2^{|\NC(\tbox)|+m}\times |\NR(\tbox)|$.

    Naturally, as a subset of $\Mos$, $\Mos'$ satisfies Conditions~\ref{it:forbidding_boolean} and~\ref{it:fresh_concepts} of Definition~\ref{def:mosaic}. The completeness of the types and the role names makes sure that Condition~\ref{it:types}. Therefore, $\Mos'$ is indeed a mosaic, and the proof is completed.
\end{proof}

\lemboundedtilesrooted*

\begin{proof}
Let $\Mos$ be a mosaic for $\Tmc$ and $\tau$ against $\emptyset$ and $Q^1$ such that tiles in $\Mos$ have (domain size) bounded by $M$, and let $m$ be the maximal size of a TQ in $Q^1$. We proceed in two stages to find a mosaic for $\Tmc$ and $\tau$ against $\emptyset$ and $Q^1$ satisfying the required bounds. 

In the first stage, we create copies $\langle \Imc_i,d_0,r\rangle$ for $0\leq i\leq m$ of each tile $\langle \Imc,d_0,r\rangle\in \Mos$, in which, with increasing $i$, the interpretation of auxiliary concept names becomes more and more `saturated.' 
To formalize this, we associate \emph{levels} with auxiliary concept names as follows. We assume that the queries in $Q^1$ use distinct variables, except the root variable $x$. We assign levels $\ell$ to variable $z$ if the distance of $x$ to $z$ is $\ell$. Hence, $x$ has level $0$ and no variable has level greater than $m-1$.

Now, the interpretation of the auxiliary concept names in copy $\langle \Imc_i,d_0,r\rangle$, $0\leq i\leq m$ is defined as follows: 
\begin{align*}
    A_{p(z)}^{\Imc_i}=\begin{cases} A_{p(z)}^\Imc & \text{if level of $z$ is at least $i$} \\
    \Delta^\Imc & \text{otherwise}
    \end{cases}
\end{align*}
Intuitively, for $i\geq 1$, the copies $\langle \Imc_i,d_0,r\rangle$ and $\langle \Imc_{i-1},d_0,r\rangle$ of $\langle \Imc,d_0,r\rangle$ is responsible for refuting only the subqueries of queries in $Q^1$ starting in level $i$.
Let $\Mos'$ be the set of all the defined tiles. We convince ourselves that $\Mos'$ is still a mosaic for $\Tmc$ and $\tau$ and against $\emptyset$ and $Q^1$. 

Note first that the copies $\langle \Imc_1,d_0,r\rangle$ are actually equal to the original mosaic $\langle\Imc,d_0,r\rangle$. Hence, Condition~3 is not affected. It is a direct consequence of the definition of the saturation (and the intuition provided above) that Condition~2 is also preserved. Finally, Condition~4 is still satisfied as well. To see it, let $\langle \Imc_i,d_0,r\rangle$ be a copy of $\langle\Imc,d_0,r\rangle$. Let $d\in \Delta^{\Imc_i}=\Delta^{\Imc}$ and $s\neq r$ or both $s$ and $r$ are non-transitive and $d\neq d_0$. Then Condition~4 in $\Mos$ was witnessed by some mosaic $\langle \Jmc,e_0,s\rangle$. It can be seen that its copy $\langle \Jmc_j,e_0,s\rangle\in\Mos'$ with $j=\min(m,i+1)$ witnesses Condition~4 in $\Mos'$. 

\smallskip In the second stage, we pick mosaics from $\Mos'$ that form a mosaic as well, as follows: 
\begin{itemize}

  \item For $i=0$, we pick tiles of shape $\langle\Imc_0,d_0,r\rangle$ from $\Mos'$ witnessing Condition~3.
  
  \item For $1\leq i<m$ we assume that tiles from level $i-1$ have already been picked and we assume that all picked tiles have shape $\langle\Imc_{i-1},d_0,r\rangle$. For each picked tile, we pick tiles of shape $\langle \Imc_i,e_0,r\rangle$ witnessing Condition~4 from as described above. 

  \item For $i=m$, we pick for each type $\tau'\subseteq \NC(\Tmc)$ and each $r\in\NR(\Tmc)$ one tile $\langle \Imc_m,d,r\rangle\in\Mos'$ (if it exists) with $\tp(\Imc_m,d)=\tau'$.
  
\end{itemize}
Note that, since we picked from a mosaic a subset of tiles satisfying Conditions~3 and~4 and Condition~2 is preserved under taking subsets, the resulting set of tiles is a mosaic as well. 

Overall, we pick $|\NR(\Tmc)|$ tiles in level $0$, and for each tile picked in level $i-1$, we pick at most $M\times |\NR(\Tmc)|$ tiles in level $i$. Finally, we pick $2^{|\NC(\Tmc)|}\times |\NR(\Tmc)|$ tiles in level $m$. Overall, this results in 
\[|\NR(\Tmc)|\cdot (M\cdot |\NR(\Tmc)|)^m+ 2^{|\NC(\Tmc)|}\cdot |\NR(\Tmc)|\]
many tiles. This is bounded by 
\[\|\Tmc\|\cdot ((M\cdot \|\Tmc\|)^{m+1}+2^{\|\Tmc\|})\]
as required.
\end{proof}

\end{document}